\documentclass[]{article}

\usepackage{primearxiv}
\usepackage[T1]{fontenc}    
\usepackage{hyperref}       
\usepackage{url}            
\usepackage{booktabs}       
\usepackage{amsfonts}       
\usepackage{nicefrac}       
\usepackage{microtype}      
\usepackage{lipsum}
\usepackage{graphicx}       
\usepackage{natbib}
\usepackage{multirow}
\usepackage{amssymb}
\usepackage{ dsfont }
\usepackage{array}
\usepackage{float}
\usepackage{algorithmic}
\usepackage{amsmath}
\DeclareMathOperator{\KL}{KL}
\usepackage{algorithm}

\newcommand*{\cD}{\mathcal{D}}
\newcommand*{\eps}{\varepsilon}
\DeclareMathOperator{\TV}{TV}
\usepackage{amsthm}
\newtheorem{theorem}{Theorem}[section]
\newtheorem{lemma}[theorem]{Lemma}
\usepackage{url}
\usepackage{xcolor}

\newtheorem{proposition}[theorem]{Proposition}

\newtheorem{definition}[theorem]{Definition}
\newtheorem{remark}[theorem]{Remark}

\newcommand{\mysec}[1]{Sec.~\ref{sec:#1}}

\newcommand{\myfig}[1]{Figure~\ref{fig:#1}}
\newcommand{\myalgo}[1]{Alg.~\ref{algo:#1}}
\newcommand{\myapp}[1]{App.~\ref{sec:#1}}

\title{Sequential Algorithms for Testing Identity and Closeness of Distributions}

\author{
	Omar Fawzi\\
	Univ Lyon, ENS Lyon, UCBL\\ CNRS, Inria, LIP, F-69342\\ Lyon Cedex 07, France\\
	\texttt{omar.fawzi@ens-lyon.fr} \\
	\And
	Nicolas Flammarion \\
	EPFL \\
	Lausanne, Switzerland\\
	\texttt{nicolas.flammarion@epfl.ch} \\
	\And
	Aurélien Garivier\\
	UMPA UMR 5669 and LIP UMR 5668 CNRS, \\
	ENS de Lyon, UCB Lyon 1\\
	Lyon, France\\  \texttt{aurelien.garivier@ens-lyon.fr}
	\And
	Aadil Oufkir\\
	  LIP UMR 5668 CNRS\\
	ENS de Lyon, UCB Lyon 1\\
	Lyon, France
	\\\texttt{aadil.oufkir@ens-lyon.fr} 
}

\begin{document}
  \maketitle
  
\begin{abstract}
What advantage do \emph{sequential} procedures provide over batch algorithms for testing properties of unknown distributions? Focusing on the problem of testing whether two distributions $\mathcal{D}_1$ and $\mathcal{D}_2$ on $\{1,\dots, n\}$ are equal or $\eps$-far, we give several answers to this question. 
We show that for a small alphabet size $n$, there is a sequential algorithm that outperforms any batch algorithm by a factor of at least $4$ in terms sample complexity. For a general alphabet size $n$, we give a sequential algorithm that uses no more samples than its batch counterpart, and possibly fewer if the actual distance $\TV(\mathcal{D}_1, \mathcal{D}_2)$ between $\mathcal{D}_1$ and $\mathcal{D}_2$ is larger than $\eps$. As a corollary, letting $\eps$ go to $0$, we obtain a sequential algorithm for testing closeness when no a priori bound on $\TV(\mathcal{D}_1, \mathcal{D}_2)$ is given that has a sample complexity $\tilde{\mathcal{O}}(\frac{n^{2/3}}{\TV(\mathcal{D}_1, \mathcal{D}_2)^{4/3}})$: this improves over the $\tilde{\mathcal{O}}(\frac{n/\log n}{\TV(\mathcal{D}_1, \mathcal{D}_2)^{2} })$ tester of~\cite{daskalakis2017optimal} and is optimal up to multiplicative constants.
We also establish limitations of sequential algorithms for the problem of testing identity and closeness: they can improve the worst case number of samples by at most a constant factor. 
\end{abstract}

\section{Introduction}\label{intro}

How to test if two discrete sources of randomness are similar or distinct?
This basic and ubiquitous question is surprisingly not closed if frugality matters, that is if one wants to take the right decision using as few samples as possible.

To state the problem more precisely, one first needs to define what ``distinct" means. In this paper, we endow the set of probability distributions on $\{1,\dots,n\}$ with the \emph{total variation distance} $\TV$, and we fix a tolerance parameter $\eps \in [0,1]$.
We consider two distributions $\mathcal{D}_1$ and $\mathcal{D}_2$, and we assume that either $\mathcal{D}_1=\mathcal{D}_2$ or $\TV(\mathcal{D}_1,\mathcal{D}_2)>\eps$. Whenever $0 < \TV(\mathcal{D}_1,\mathcal{D}_2)\leq \eps$, we do not expect any determined behaviour from our test. 
Two cases occur: 
\begin{itemize}
    \item when the first distribution $\mathcal{D}_1$ is fixed and known to the algorithm (but not $\mathcal{D}_2$), we say that we are \emph{testing identity} using independent samples of $\mathcal{D}_2$; 
    \item when both distributions are unknown we are \emph{testing closeness}, based on an equal number of independent samples of both distributions.
\end{itemize}

We also need to specify what kind of ``test" is considered. Here we treat the two hypotheses symmetrically (there is no ``null hypothesis") 
: given a fixed risk $\delta\in(0,1)$, we expect our procedure to find the true one with probability $1-\delta$, whichever it is. We call such a procedure \emph{$\delta$-correct}.

Finally, we consider and compare two notions of ``frugality": in the \emph{batch} setting, the agent specifies in advance the number of samples needed for the test: she takes her decision just after observing the data all at once, and the sample complexity of the test is the smallest sample size of a $\delta-$correct procedure. In the \emph{sequential} setting, the agent observes the samples one by one, and decides accordingly whether she takes her decision or requests to see more samples before making a decision. Then, the sample complexity of the test is the smallest \emph{expected number of samples} needed before a $\delta$-correct procedure takes a decision. Note that this expected number depends on the unknown distributions $\mathcal{D}_1$ and $\mathcal{D}_2$, which turns out to be an important advantage of sequential procedures.

\paragraph{Contributions} 
 When $n \ge 2$ is small, we show that the optimal sample complexities can be precisely characterized (up to lower order terms in $\eps$) in both the batch and sequential setting as shown in Table~\ref{tab:n=2_summary} and Table~\ref{tab:n=2_summary2}. This establishes a provable advantage for sequential strategies over batch strategies when $n\ll \log(1/\delta)$: sequential algorithms reduce the sample complexity by a factor of at least $4$, and can stop rapidly if the tested distributions are far (i.e., $\TV(\mathcal{D}_1,\cD_2) > \eps$). The improvements of the sequential algorithm are illustrated in Fig.~\ref{fig:synthetic}. The sequential algorithms use stopping rules inspired by time uniform concentration inequalities. The problems of testing identity and closeness for small $n$ are studied in \mysec{bernoulli-id} and \mysec{bernoulli} respectively.  
\begin{table}[t!]
\centering
\begin{tabular}{  c  c c} 
\hline
  \textbf{Model} &Lower bound  &Upper bound \\
  \hline
  Batch  & $8\frac{\lfloor n^2/4\rfloor}{n^2}\log(1/\delta)\eps^{-2}-\mathcal{O}\left(n\log\log(1/\delta)\eps^{-2}\right) $ &$8\frac{\lfloor n^2/4\rfloor}{n^2}\log(1/\delta)\eps^{-2}+8\frac{\lfloor n^2/4\rfloor}{n^2}(n+1)\eps^{-2}$ 
  \\ 
  \hline
  \multirow{2}{*}{Sequential ($\tau_1)$}    &\multirow{2}{*}{ $2\frac{\lfloor n^2/4\rfloor}{n^2}\log(1/\delta)\eps^{-2} -\mathcal{O}(\eps^{-2})  $}&$2\frac{\lfloor n^2/4\rfloor}{n^2}\log(1/\delta)\eps^{-2}$  \\ 
  &&$+\mathcal{O}\left((n+\log(1/\delta)^{2/3})\eps^{-2}\right)$\\
  \hline
  \multirow{2}{*}{Sequential ($\tau_2)$}    & \multirow{2}{*}{  $2\frac{|B_{opt}|}{n}\left(1-\frac{|B_{opt}|}{n}\right)\log(1/\delta)d^{-2} -\mathcal{O}(d^{-2})$   }& $  2\frac{|B_{opt}|}{n}\left(1-\frac{|B_{opt}|}{n}\right)\log(1/\delta)d^{-2}   $   \\
  &&$+\mathcal{O}\left((n+\log(1/\delta)^{2/3})d^{-2}\right)  $ \\
  \hline
\end{tabular}
\caption{Lower and upper bounds on  sample complexity for testing uniform in batch and sequential setting with $d=\TV(\mathcal{D},U_n)=|\mathcal{D}(B_{opt})-|B_{opt}|/n|$.  $\tau_1$ (resp. $\tau_2$) represents the stopping time of the sequential algorithm when $\cD=U_n$ (resp. $\TV(\mathcal{D},U_n)>\eps$). The $\mathcal{O}$ hides universal constants.
}
\label{tab:n=2_summary}
\end{table}
\begin{table}[t!]
\centering
\begin{tabular}{  c  c c} 
\hline
  \textbf{Model} &Lower bound  &Upper bound \\
  \hline\hline
  Batch  & $4\log(1/\delta)\eps^{-2}-\mathcal{O}(\log\log(1/\delta)\eps^{-2}) $ &$4\log(1/\delta)\eps^{-2} + \mathcal{O}(n\eps^{-2})$ 
  \\ 
  \hline
   \multirow{2}{*}{Sequential ($\tau_1)$}    & \multirow{2}{*}{$\log(1/\delta)\eps^{-2} - \mathcal{O}(\eps^{-2}) $ }&$\log(1/\delta)\eps^{-2}$  \\ 
  &&$+\mathcal{O}( (n + \log(1/\delta)^{2/3})\eps^{-2})$\\
  \hline
  \multirow{2}{*}{Sequential ($\tau_2)$} &  \multirow{2}{*}{$\log(1/\delta)d^{-2} - \mathcal{O}(d^{-2}) $ }  & $  \log(1/\delta)d^{-2}  $   \\
  &&$+\mathcal{O}\left((n + \log(1/\delta)^{2/3})d^{-2}\right)   $ \\
  \hline
\end{tabular}
\caption{Lower and upper bounds on  the sample complexities for testing closeness in the batch and sequential settings with $d=\TV(\mathcal{D}_1,\mathcal{D}_2).$ $\tau_1$ (resp. $\tau_2$) represents the stopping time of the sequential algorithm when $\cD_1=\cD_2$ (resp. $\TV(\mathcal{D}_1,\cD_2)>\eps$).
The $\mathcal{O}$ hides universal constants. 
}
\label{tab:n=2_summary2}
\end{table}
\\

For general $n\ge 2$, 
we improve the dependence on $\eps$ to $\eps \vee \TV(\mathcal{D}_1,\mathcal{D}_2)$ in the best batch algorithm due to \cite{diakonikolas2020optimal}, which is known to be optimal up to multiplicative constants. Namely we obtain a sequential closeness testing algorithm using a number of samples given by
\begin{equation}
\mathcal{O}\left( \max\left( \frac{n^{2/3}\log^{1/3}(1/\delta)}{(\eps\vee \TV(\mathcal{D}_1,\mathcal{D}_2))^{4/3}},\frac{n^{1/2}\log^{1/2}(1/\delta)}{(\eps\vee \TV(\mathcal{D}_1,\mathcal{D}_2))^2},\frac{\log(1/\delta)}{(\eps\vee \TV(\mathcal{D}_1,\mathcal{D}_2))^2}\right)\right)\;.
\label{eq:testing_closeness_samples}
\end{equation}
A doubling search technique could also lead to the same order of sample complexity, we explain this method and compare it with our proposed algorithm in Rem.~\ref{rem:doubling}.

As a special case, when $\eps=0$ (the algorithm should not stop when $\mathcal{D}_1 = \mathcal{D}_2$ in this case) we show that there is an algorithm that stops after \begin{equation}
    \mathcal{O}\left(\max\left( \frac{\log\log(1/d)}{d^2},\frac{n^{2/3}\log\log(1/d)^{1/3}}{d^{4/3}} ,\frac{n^{1/2}\log\log(1/d)^{1/2}}{d^{2}}\right) \right)
    \label{eq:testing_diff_samples}
\end{equation}
samples where $d=\TV(\mathcal{D}_1,\mathcal{D}_2) > 0$. This is an improvement over the sequential algorithm of \cite{daskalakis2017optimal} which uses $\Theta(\frac{n/\log n}{d^2}\log \log (1/d))$ samples. 
We design the stopping rules according to a time uniform concentration inequality deduced from McDiarmid's inequality, where we use the ideas of \cite{howard2018uniform,howard2020time} in order to obtain powers of $\log\log(1/d)$ instead of  $\log(1/d)$. 

We show that the sample complexity for the testing closeness problem given by Eq.~\eqref{eq:testing_closeness_samples} is optimal up to multiplicative constants in the worst case setting (i.e., when looking for a bound independent of the distributions $\mathcal{D}_1$ and $\mathcal{D}_2$). To do so, we construct two families of distributions whose cross TV distance is exactly $d\ge \eps$ and hard to distinguish unless we have a number of samples given by Eq.~\eqref{eq:testing_closeness_samples}. This latter  lower bound is based on properties of KL divergence along with Wald's Lemma. Using similar techniques, we also establish upper and lower bounds for testing identity that match up to multiplicative constants.

In addition, we establish a lower bound on the number of queries that matches Eq.~\eqref{eq:testing_diff_samples} up to multiplicative constants. 
The proof is inspired by  \cite{karp2007noisy} who proved lower bounds for testing whether the mean of a sequence of i.i.d. Bernoulli variables is smaller or larger than $1/2$. We construct well-chosen distributions $\cD_k$ (for $k$ integer)  that are at distance $\eps_k$ ($\eps_k$ decreasing to $0$) from uniform and then use properties of the Kullback-Leibler’s divergence to show that no algorithm can distinguish between $\cD_k$ and uniform using fewer samples than in Eq.~\eqref{eq:testing_diff_samples}. Note that we could have used the testing closeness lower bound described in the previous paragraph and let $\eps=0$, however this gives sub-optimal lower bounds. 
 \begin{figure}[t]
\centering
\begin{minipage}[c]{.5\linewidth}
\includegraphics[width=1\linewidth]{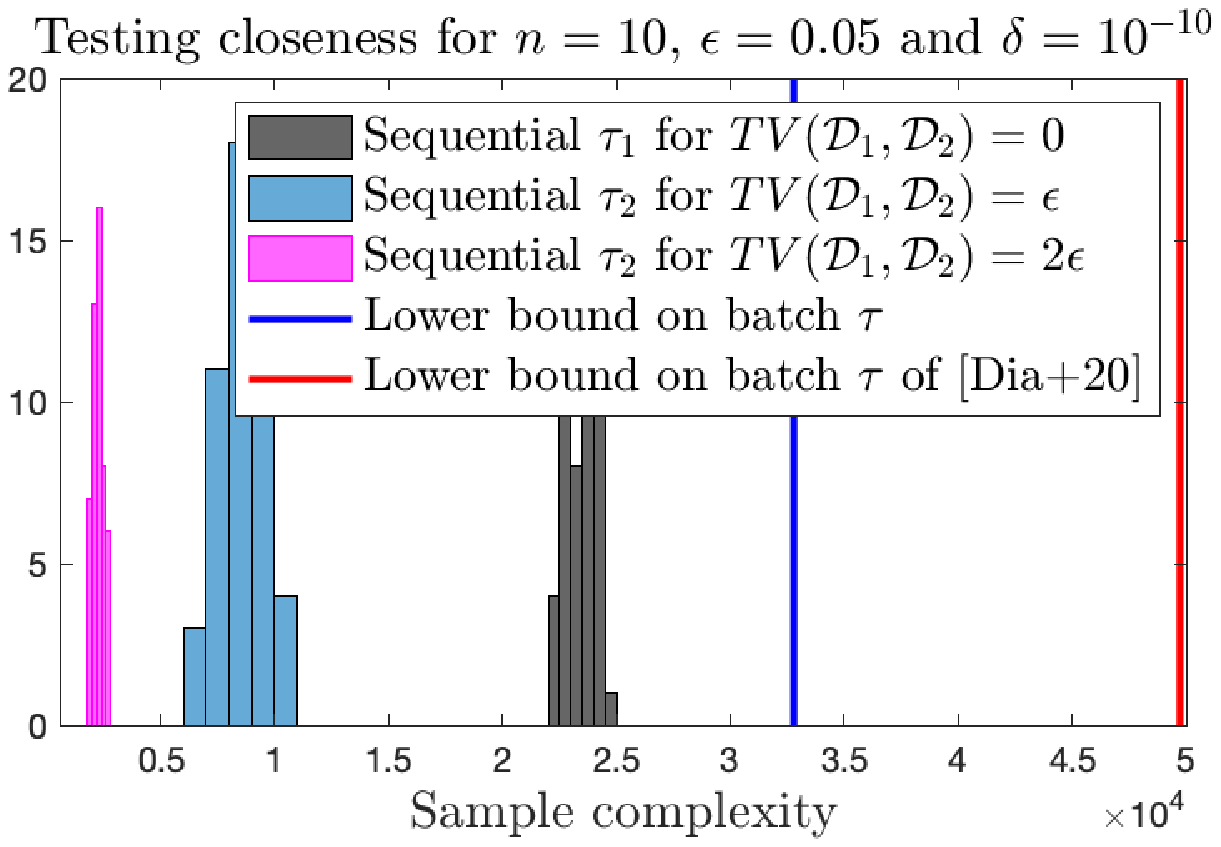}
   \end{minipage}
   \hspace*{-10pt}
   \begin{minipage}[c]{.5\linewidth}
\includegraphics[width=1\linewidth]{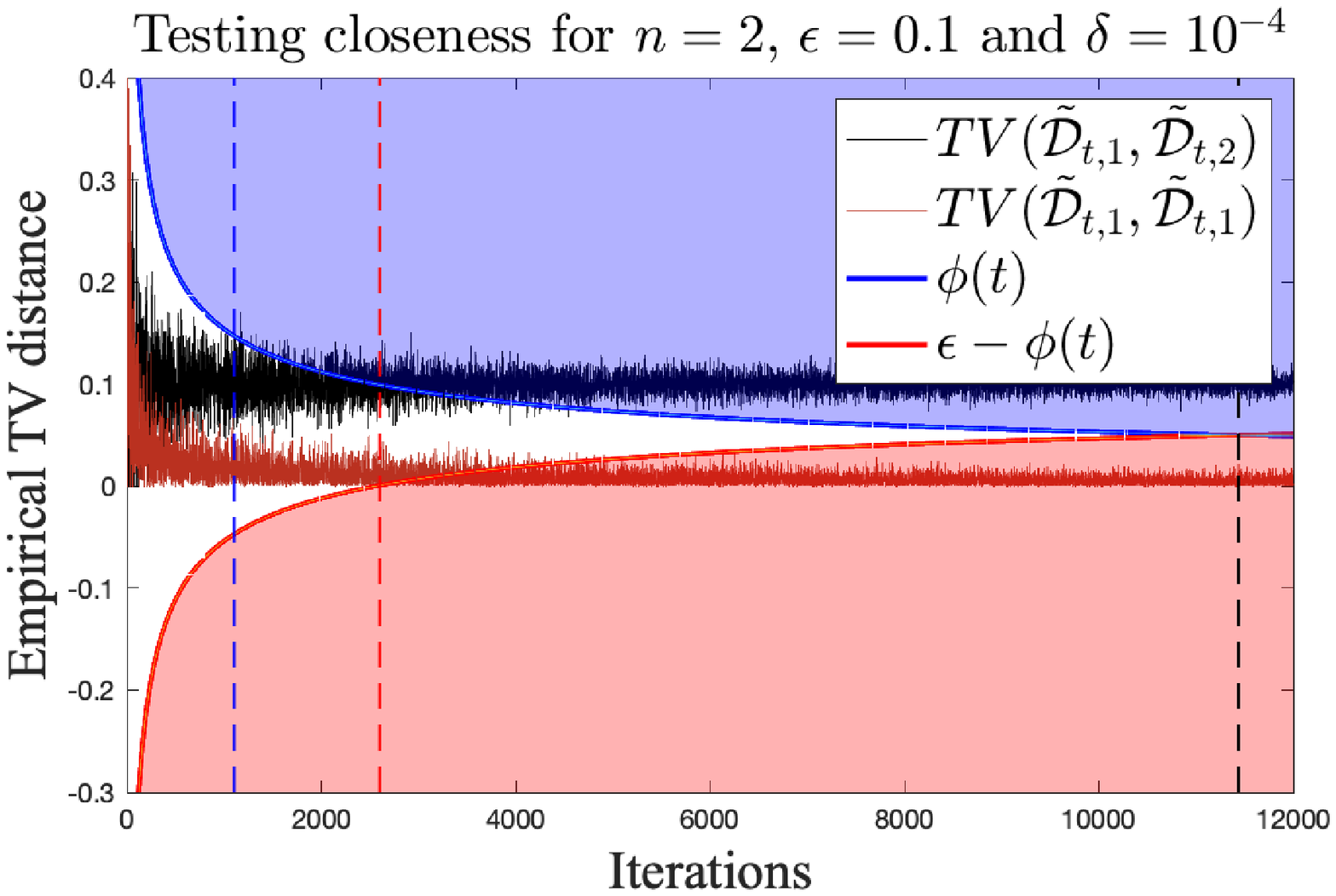}
   \end{minipage}
  \caption{Left: histogram of the stopping times for $100$ Monte-Carlo experiments. Black: ${\cal D}_1={\cal D}_2=U_n$, blue (resp. magenta): $\mathcal{D}_1=U_n$ and $\mathcal{D}_2=\{(1\pm 2\eps)/n\}$  (resp. $\{(1\pm 4\eps)/n\})$. Right: $\mathcal{D}_1=U_2$ and $\mathcal{D}_2=\{(1\pm 2\eps)/2\}$. The sequential tester stops as soon as the statistic enters the red region (for $H_1$) or blue region (for $H_2$) whereas the batch tester waits for the red and blue regions to cover the whole segment $[0,1]$. The blue/red and black dashed lines represent respectively the stopping times of the sequential and batch algorithms. We note that, in both cases, the sequential tester stops long before the batch algorithm.}
     \label{fig:synthetic}
\end{figure}

\paragraph{Discussion of the setting and related work}

It is clearly impossible to test $\mathcal{D}_1=\mathcal{D}_2$ versus $\mathcal{D}_1\neq \mathcal{D}_2$ in finite time: this is why the slack parameter $\eps$ is introduced in this setting. Other authors like~\cite{daskalakis2017optimal} make a different choice: they fix no $\eps$, but only require that the test decides for $\mathcal{D}_1\neq \mathcal{D}_2$ as soon as it can, and never stops with high probability when $\mathcal{D}_1 = \mathcal{D}_2$.

We focus on the TV distance in testing closeness problems because it 
characterises the  probability of error for  the problem of distributions discrimination 
; as noted by \cite{MR3775963}, using other distances such as $\KL$ and $\chi^2$ is in general impossible.

For an overview of testing discrete distributions we recommend the survey of \cite{canonne2020survey}. Testing identity for the uniform distribution was solved by  \cite{paninski2008coincidence}, then for general distribution by \cite{valiant2017automatic} and finally the high probability version by \cite{diakonikolas2017optimal}. Likewise testing closeness was solved by \cite{chan2014optimal}, and a distribution dependent complexity was found by \cite{diakonikolas2016new} and finally the  high probability version by \cite{diakonikolas2020optimal}. Besides, the problem of testing $\mathcal{D}_1=\mathcal{D}_2$ vs $\mathcal{D}_1\neq\mathcal{D}_2$ was solved by \cite{daskalakis2017optimal} for $n=2$, however the constants are not optimal. They also propose algorithms for the general case using black-box reduction from non-sequential hypothesis testers. Sequential and adaptive procedures have also been explored in active hypothesis setting \citep{naghshvar2013sequentiality}, channels' discrimination \citep{hayashi2009discrimination} and quantum hypothesis testing \citet{li2021optimal}. Sequential strategies have been also considered for testing continuous distributions by \cite{zhao2016adaptive} and \cite{balsubramani2015sequential}. In the latter, the authors design sequential algorithms whose stopping time adapts to the unknown difficulty of the problem. The techniques used are time uniform concentration inequalities which are surveyed by \cite{howard2020time}. In contrast to the present work, however, they test properties of the \emph{means} of the distributions.  
%

\section{Preliminaries}
We mostly follow \cite{daskalakis2017optimal} for the notation.
\subsection{Testing identity}
 Given two distributions $\mathcal{D}$ (known)  and $\mathcal{D}'$ (unknown) on $[n] := \{1,\dots,n\}$, we want to distinguish between two hypothesis  $H_1 :\mathcal{D}'=\mathcal{D}$ and $H_2: \TV(\mathcal{D}',\mathcal{D})>\eps$ where the \textit{Total Variation} 
is defined as  $
 \TV(\cD,\cD')=\frac{1}{2}\sum_{i=1}^n |\cD_i-\cD'_i|.$ 
 We call a stopping rule a function $T:[n]^* \rightarrow \{0,1,2\}$ such that if $T(x)\neq 0$ then $T(xy)=T(x)$ for all strings $x$ and $y$.  $T(x)=1$ (resp. $T(x)=2$) means that the rule accepts $H_1$ (resp. $H_2$) after seeing $x$ while $T(x)=0$ means the rule does not make a choice and continues sampling. 
 We define two different stopping times, the first $\tau_1(T,\mathcal{D}')=\inf\{t, T(x_1\cdots x_t)=1\}$ and the second $\tau_2(T,\mathcal{D}')=\inf\{t, T(x_1\cdots x_t)=2\}$ where $x_1,\dots $ are i.i.d. samples from $\mathcal{D}'$. We want to find stopping rules satisfying 
 \begin{enumerate}
\item   $ \mathds{P}\left(\tau_2(T,\mathcal{D})\le\tau_1(T,\mathcal{D})\right) \le \delta$ and    
\item  $ \mathds{P}\left(\tau_1(T,\mathcal{D}')\le \tau_2(T,\mathcal{D}')\right) \le \delta$ whenever $\TV(\mathcal{D}',\mathcal{D})>\eps$.
 \end{enumerate}
We call such a stopping rule $\delta$-correct.  Our goal is to minimize the expected sample complexity $\mathds{E}(\tau_1(T,\mathcal{D}))$ in case of the input is from $\mathcal{D}$ and $\mathds{E}(\tau_2(T,\mathcal{D}'))$  in case of the input is from $\mathcal{D}'$ such that $\TV(\mathcal{D}',\mathcal{D})>\eps$.

A batch algorithm is one for which $\tau = \tau_1 = \tau_2$ is a constant random variable which only depends on $\delta, \eps, n$ and $\mathcal{D}$. 

 \subsection{Testing closeness}
 Given two distributions $\mathcal{D}_1$ and $\mathcal{D}_2$ on $\{1,\dots,n\}$ we want to distinguish between two hypothesis  $H_1: \mathcal{D}_1=\mathcal{D}_2$ and $H_2: \TV(\mathcal{D}_1,\mathcal{D}_2)>\eps$. 
 We call a stopping rule a function $T:\bigcup_{k\in\mathds{N}}[n]^k\times [n]^k \rightarrow \{0,1,2\}$ such that if $T(x,y)\neq 0$ then $T(xz,yt)=T(x,y)$ for all strings $x,y,z,t$ with $|x|=|y|$ and $|z|=|t|$.  $T(x,y)=1$ (resp. $T(x,y)=2$) means that the rule accepts $H_1$ (resp. $H_2$) after seeing the sequences $x$ and $y$ while $T(x,y)=0$ means the rule doesn't make a choice and continue sampling. We define two different stopping times, the first $\tau_1(T,\mathcal{D}_1,\mathcal{D}_2)=\inf\{t, T(x_1\cdots x_t,y_1\dots y_t)=1\}$ and the second $\tau_2(T,\mathcal{D}_1,\mathcal{D}_2)=\inf\{t, T(x_1\cdots x_t,y_1\dots y_t)=2\}$ where $x_1,\dots $ are i.i.d. samples from $\mathcal{D}_1$ and $y_1,\dots $ samples from $\mathcal{D}_2$ . We want to find stopping rules satisfying  
  \begin{enumerate}
 \item  $ \mathds{P}\left(\tau_2(T,\mathcal{D}_1,\mathcal{D}_2)\le \tau_1(T,\mathcal{D}_1,\mathcal{D}_2)\right) \le \delta$  if  $\mathcal{D}_1=\mathcal{D}_2$ and    
\item  $ \mathds{P}\left(\tau_1(T,\mathcal{D}_1,\mathcal{D}_2)\le \tau_2(T,\mathcal{D}_1,\mathcal{D}_2)\right) \le \delta$ whenever $\TV(\mathcal{D}_1,\mathcal{D}_2)>\eps$.
 \end{enumerate}
We call such a stopping rules $\delta$-correct. Our goal is to minimize the expected sample complexity $\mathds{E}(\tau_1(T,\mathcal{D}_1,\mathcal{D}_2))$ in case of the input is from $\mathcal{D}_1$, $\mathcal{D}_2$ such that $\mathcal{D}_1=\mathcal{D}_2$ and $\mathds{E}(\tau_2(T,\mathcal{D}_1,\mathcal{D}_2))$  in case of the input is from $\mathcal{D}_1$, $\mathcal{D}_2$ such that $\TV(\mathcal{D}_1,\mathcal{D}_2)>\eps$.

\section{Testing identity for small $n$\label{sec:bernoulli-id}}

In this section, we focus on small $n\ge2$ and we consider two distributions $\mathcal{D}=U_n$ and $\mathcal{D}'$ on $[n]$.  In this case, the hypothesis $H_1$ becomes $\mathcal{D}'=U_n$ and $H_2$ becomes $\TV(\mathcal{D}',U_n)>\eps$. We are interested in precisely comparing the sample complexity of testing identity in the sequential versus the batch setting. 
In order to find the optimal constant, we first need to obtain a sharp lower bound in the batch setting, which is done directly by using Stirling's approximation. We then turn to the sequential case.



\subsection{Batch setting}\label{sec:test-iden-be}


In the batch setting, the number of steps $\tau $ is fixed before the test. The tester samples $A_1,\dots,A_\tau \sim \mathcal{D}'$ and decides according to the comparison between the empirical TV distance  $\TV(\tilde{\mathcal{D}}'_\tau,U_n)$ and $\eps/2$ where $\tilde{\mathcal{D}}'_{\tau}=\left\{\left(\sum_{j=1}^{\tau} 1_{A_j= i}\right)/\tau\right\}_{i\in[n]}$. If  $\TV(\tilde{\mathcal{D}}'_\tau,U_n)\le \eps/2$  she accepts $H_1$  and rejects it otherwise. In order to control the number of steps $\tau$ so that the error of this algorithm does not exceed $\delta$, Chernoff–Hoeffding's inequality~(\cite{hoeffding1994probability}) writes for i.i.d. random variables $X_1,\dots,X_\tau\sim \mathcal{B}(q)$:
\begin{align}\label{eq:hoefdding}\tag{C-H}
\mathds{P}\!\left(\frac{\sum_{i=1}^{\tau}X_i}{\tau}\!-\!q>\frac{\eps}{2}\right)\!\le\! e^{-\tau\KL(q+\eps/2 ,  q) } \text{ and }
\mathds{P}\!\left(\frac{\sum_{i=1}^{\tau}X_i}{\tau}\!-\!q<\!-\frac{\eps}{2}\right)\!\le\! e^{-\tau\KL(q-\eps/2 , q) }\;.
\end{align}
We use the following property of $\TV$ distance:
 \begin{align*}
 \TV(\mathcal{D}',U_n)&=\max_{B\subset[\lfloor n/2 \rfloor ]} |\mathcal{D}'(B)-|B|/n|=|\mathcal{D}'(B_{opt})-|B_{opt}|/n|,
 \end{align*}and choose $X_i=1_{A_i\in B_{opt}}\sim \mathcal{B}(\mathcal{D}'(B_{opt}))$.

Applying these inequalities for $\mathcal{D}'=U_n$ (to control the type I error) 
 and  for $\mathcal{D}'\neq U_n$ (to control  the type II error)
prove that this test is $\delta$-correct if 
\begin{align*}
\tau=\max_{b\in[n]}\left\{\frac{\log(2/\delta)}{\KL(b/n\pm\eps/2,b/n\pm\eps)}, \frac{\log(2^{n+1}/\delta)}{\KL(b/n\pm\eps/2,b/n)} \right\},
\end{align*}
where $\KL(p,q)=\KL(\mathcal{B}(p),\mathcal{B}(q))$ denotes the Kullback-Leibler divergence.

To see this, we analyze the three cases: $\mathcal{D}'=U_n$, $\mathcal{D}'(B_{opt})-|B_{opt}|/n>\eps$ and $\mathcal{D}'(B_{opt})-|B_{opt}|/n<-\eps$. They are all handled by a simple application of Chernoff–Hoeffding's inequality~\eqref{eq:hoefdding}:
\begin{itemize}
\item If $\mathcal{D}'=U_n$, the probability of error is given by \begin{align*}
\mathds{P}\left(\left|\TV(\tilde{\mathcal{D}}',U_n)\right|>\frac{\eps}{2}\right)&=\mathds{P}\left(\exists B\subset [n]: \left|\tilde{\mathcal{D}}'(B)-|B|/n\right|>\frac{\eps}{2}\right)
\\&\le  \sum_{B\subset[\lfloor n/2 \rfloor ]} e^{-\tau\KL(|B|/n+\eps/2|||B|/n)}+e^{-\tau\KL(1-|B|/n+\eps/2||1-|B|/n)}
         \\&\le \delta\;.
\end{align*}
\item If $\mathcal{D}'(B_{opt})-|B_{opt}|/n>\eps$, the probability of error is given by \begin{align*}
\mathds{P}\left(\left|\TV(\tilde{\mathcal{D}}',U_n)\right|\le\frac{\eps}{2}\right)&\le  \mathds{P}\left(\tilde{\mathcal{D}}'(B_{opt})-|B_{opt}|/n\le \frac{\eps}{2}\right)
           \\&\le e^{-\tau\KL(|B_{opt}|/n+\eps/2||\mathcal{D}'(B_{opt}))}
\\&\le e^{-\tau\KL(|B_{opt}|/n+\eps/2|||B_{opt}|/n+\eps)} \text{ because } x \mapsto \KL(p||x) \text{ is increasing on } (p,1)
			\\&\le \delta\;.
\end{align*}
\item If $\mathcal{D}'(B_{opt})-|B_{opt}|/n<-\eps$, the probability of error is given by \begin{align*}
\mathds{P}\left(\left|\TV(\tilde{\mathcal{D}}',U_n)\right|\le\frac{\eps}{2}\right)&\le  \mathds{P}\left(\tilde{\mathcal{D}}'(B_{opt})-|B_{opt}|/n\ge -\frac{\eps}{2}\right)
           \\&\le e^{-\tau\KL(|B_{opt}|/n-\eps/2||\mathcal{D}'(B_{opt}))}
\\&\le e^{-\tau\KL(|B_{opt}|/n-\eps/2|||B_{opt}|/n-\eps)} \text{ because } x\rightarrow\KL(p||x) \text{ is decreasing on } (0,p)
			\\&\le \delta\;.
\end{align*}
\end{itemize}

We show in the following theorem that this number of steps $\tau$ is necessary. 
\begin{theorem}\label{ber-id-ns}
In the batch setting, any $\delta$-correct algorithm testing identity to $\mathcal{D}=U_n$ requires at least $\tau$ samples, where
\begin{align*}
\tau\ge\max_{b\in[n]}\min\left\{\frac{\log(1/\delta)}{\KL(b/n+\eps/2,b/n+\eps)}, \frac{\log(1/\delta)}{\KL(b/n+\eps/2,b/n)} \right\}   -\mathcal{O}\left(\frac{n\log\log1/\delta}{\eps^2}\right) \;.
\end{align*}
\end{theorem}

This lower bound has the simple equivalent $8\frac{\lfloor n^2/4\rfloor}{n^2}\log(1/\delta)\eps^{-2}-\mathcal{O}\left(n\log\log(1/\delta)\eps^{-2}\right)$ when $\eps\rightarrow 0$ (see Lemma \ref{lemma-kl} for the equivalent of $\KL$ divergence). 
To prove this lower bound, we show that every $\delta$-correct tester can be transformed into a test which depends only on the numbers of $1's, 2's,\dots,n's$ occurred on $\{A_1,\dots,A_\tau\}$ . We then consider the distribution $\mathcal{D'}$  with roughly half parts are  $1/n+\eps/\lfloor n/2\rfloor$ and the others are $1/n-\eps/\lceil n/2\rceil$and derive tight lower bounds on the probability mass function of the multinomial distribution. 
\begin{proof}
We consider such a $\delta$-correct test $A:\{1,\dots,n\}^\tau \rightarrow \{0,1\}$, it sees a word consisting of $\tau$ samples either from a distribution $\eps$-far from $U_n$ or $U_n$ and returns $1$ if it thinks the the samples come from $U_n$ and $2$ otherwise. We construct another test $B:\{1,\dots,n\}^\tau \rightarrow \{0,1\}$by the expression 
\begin{align*}
B(x)=1_{\sum_{\sigma\in \mathcal{S}_\tau}A(\sigma(x))-1 \ge \tau!/2 } \;,
\end{align*}
$B$ can be proven to be $2\delta$-correct and have the property of invariance under the action of the symmetric group.  
Let $d\in[n]$ and consider the word 
\begin{align*}
w=1^{\tau\left(\frac{1}{n}+\frac{\eps}{2d}\right)}\dots d^{\tau\left(\frac{1}{n}+\frac{\eps}{2d}\right)} (d+1)^{\tau\left(\frac{1}{n}-\frac{\eps}{2(n-d)}\right)}\dots n^{\tau\left(\frac{1}{n}-\frac{\eps}{2(n-d)}\right)},
\end{align*}
we have two choices, either $B(w)=0$ or $B(w)=1$, we suppose the first and take $q$ the distribution defined by $q_1=\dots=q_d=\frac{1}{n}+\frac{\eps}{d}$ and $q_{d+1}=\dots=q_n=\frac{1}{n}-\frac{\eps}{n-d}$. It satisfies $\TV(q,U_n)=\eps$ thus $\mathds{P}_q(x_1\cdots x_\tau=w)\le \delta$ hence
\begin{align*}
{\tau \choose \tau_1\cdots \tau_n} \left(\frac{1}{n}+\frac{\eps}{n}\right)^{d\tau_1}\left(\frac{1}{n}-\frac{\eps}{n-d}\right)^{(n-d)\tau_{d+1}} \le \delta
\end{align*}
where $\tau_1=\dots=\tau_d= \tau\left( \frac{1}{n}+\frac{\eps}{2d}\right)$  and $ \tau_{d+1}=\dots=\tau_n=\tau\left(\frac{1}{n}-\frac{\eps}{2(n-d)}\right)$ thus 
\begin{align*}
\frac{\tau !}{(\tau_1 !)^d (\tau_{d+1}!)^{n-d}}\left(\frac{1}{n}+\frac{\eps}{n}\right)^{d\tau_1}\left(\frac{1}{n}-\frac{\eps}{n-d}\right)^{(n-d)\tau_{d+1}} \le \delta,
\end{align*} which implies by Stirling's approximation 
\begin{align*}
\frac{e(\tau/e)^\tau }{(e\tau_1(\tau_1/e)^{\tau_1})^d (e\tau_{d+1}(\tau_{d+1}/e)^{\tau_{d+1}})^{n-d}}e^{-d\tau_1\log\frac{\tau_1}{q_1}}\left(\frac{1}{n}+\frac{\eps}{n}\right)^{d\tau_1}\left(\frac{1}{n}-\frac{\eps}{n-d}\right)^{(n-d)\tau_{d+1}} \le \delta,
\end{align*}after simplifying we obtain
\begin{align*}
\frac{e}{(e\tau_1)^d (e\tau_{d+1})^{n-d}}e^{-d\tau_1\log\frac{\tau_1}{q_1}}e^{-(n-d)\tau_{d+1}\log\frac{\tau_{d+1}}{q_{d+1}}}\le \delta,
\end{align*}
or \begin{align*}
\frac{e}{(e\tau_1)^d (e\tau_{d+1})^{n-d}}e^{-\tau \KL(d/n+\eps/2 , d/n+\eps)}\le \delta,
\end{align*}
Finally
\begin{align*}
\tau \ge \frac{\log(1/\delta)+1-n}{\KL(d/n+\eps/2,d/n+\eps)}-\mathcal{O}\left(n\log\log(1/\delta)\eps^{-2}\right).
\end{align*}
If $B(w)=1$, we consider $q=U_n$ and we obtain with the same approach 
\begin{align*}
\tau \ge \frac{\log(1/\delta)+1-n}{\KL(d/n+\eps/2,d/n)}-\mathcal{O}\left(n\log\log(1/\delta)\eps^{-2}\right).
\end{align*}These lower bounds work for all $d\in [n]$  therefore 

\begin{align*}
\tau \ge\max_{d\in[n]} \min \left\{ \frac{\log(1/\delta)}{\KL(d/n+\eps/2,d/n)},\frac{\log(1/\delta)}{\KL(d/n+\eps/2,d/n+\eps)}\right\}-\mathcal{O}\left(\frac{n\log\log(1/\delta)}{\eps^{2}}\right).
\end{align*}
\end{proof}

This simple analysis relies on well-known arguments for testing Bernoulli variables $\cD_1=\mathcal{B}(p)$ and $\cD_2=\mathcal{B}(q)$. For example, \cite{anthony2009neural} and \cite{karp2007noisy} test whether $q=1/2+\eps$ or $q=1/2-\eps$ with an error probability $\delta$. \cite{anthony2009neural} show that we need roughly $\log(1/\delta)\eps^{-2}/4$ samples while \cite{karp2007noisy} prove that $2\log(1/\delta)\eps^{-2}$ samples are sufficient. If $\eps$ is not known to the tester, sequential algorithms prove to be essential. Indeed, \cite{karp2007noisy} manage to prove that $\Theta(\log\log(1/|q-1/2|)|q-1/2|^{-2})$ is necessary and sufficient to test $q>1/2$ vs $q<1/2$ with an error probability $1/3$. In what follows, we use sequential algorithms to expose the dependency on $\TV(\mathcal{D}',U_n)$ for the testing identity problem.

\subsection{Sequential setting}



If one wants to leverage the sequential setting to improve the optimal sample complexity of testing identity, it is natural to first investigate how it can be improved by removing the batch assumption of the previous lower-bound in Theorem~\ref{ber-id-ns}. We first state a new lower bound inspired by the work of \cite{garivier2019non}.  
\begin{lemma}\label{lowerbound-kl-id}
Let $\mathcal{D}=U_n$ be the uniform distribution. Let $T$ a stopping rule for testing identity: $\mathcal{D}'=U_n$ vs ${\TV(\mathcal{D}',U_n)>\eps}$ with an error probability $\delta$. Let  $\tau_1$ and $\tau_2$ the associated stopping times. We have 
\begin{itemize}
\item $\mathds{E}(\tau_1(T,U_n) \ge \frac{\log(1/3\delta)}{\min_{b\in[n]}\{\KL(b/n,b/n\pm\eps)\}}$ \text{ if } $\mathcal{D}'=U_n$.
\item $\mathds{E}(\tau_2(T,\mathcal{D}'))\ge \frac{\log(1/3\delta)} {\min\{\KL(|B_{opt}|/n\pm d,|B_{opt}|/n)\}}$    \text{ if }${d=\TV(\mathcal{D}',U_n)=|\mathcal{D}'(B_{opt})-|B_{opt}|/n|>\eps}$.
\end{itemize}
\end{lemma}
An average number of samples equivalent  to $2\frac{\lfloor n^2/4\rfloor}{n^2}\log(1/3\delta)\eps^{-2}  $ (by Lemma \ref{lemma-kl})is thus necessary when the tester can access sequentially to the samples, which is roughly 4 times less than the complexity  obtained in Theorem~\ref{ber-id-ns} for the batch setting. 
The proof, with a strong information-theoretic flavor, compares two situations: when the samples are from equal distributions and when they are from $\eps$-far distributions. Those samples cannot be distinguished until their size is large enough, as can be proved by combining properties of Kullback-Leibler's divergence and Wald's lemma. 
\begin{proof}
We apply the lower bounds of \mysec{proof-lowerbound-kl-id}. If $\mathcal{D}=U_n$, we set $\mathcal{D}^+_b$ the distribution whose first $b$ parts are equal to $1/n+\eps/b$ and the others are equal to $1/n-\eps/(n-b)$. 
\begin{align*}
\mathds{E}(\tau_1(\mathcal{D})) &\ge \frac{\log1/3\delta}{\min{\mathcal{D}'' \text{s.t.}\TV(\mathcal{D}'',\mathcal{D})>\eps } \KL(\mathcal{D} , \mathcal{D}'')}
\\&\ge \frac{\log1/3\delta}{\min_{b\in[n] } \KL(\mathcal{D} , \mathcal{D}^+_b)}
\\&\ge \frac{\log1/3\delta}{\min_{b\in[n] } \KL(b/n, b/n+\eps/n)}.
\end{align*}
Likewise we can prove for $\mathcal{D}^-_b$ the distribution whose first $b$ parts are equal to $1/n-\eps/b$ and the others are equal to $1/n+\eps/(n-b)$:
\begin{align*}
\mathds{E}(\tau_1(\mathcal{D})) &\ge \frac{\log1/3\delta}{\min_{\mathcal{D}'' \text{s.t.}\TV(\mathcal{D}'',\mathcal{D})>\eps } \KL(\mathcal{D} , \mathcal{D}'')}
\\&\ge \frac{\log1/3\delta}{\min_{b\in[n] } \KL(\mathcal{D} , \mathcal{D}^-_b)}
\\&\ge \frac{\log1/3\delta}{\min_{b\in[n] } \KL(b/n, b/n-\eps/n)}.
\end{align*}
Finally:
\begin{align*}
\mathds{E}(\tau_1(\mathcal{D})) \ge \frac{\log1/3\delta}{\min_{b\in[n] } \KL(b/n,b/n\pm\eps/n)}.
\end{align*}
Now, in order to prove a lower bound on $\tau_2$, we focus on distributions that are $\eps$-far from $U_n$ and have  the same length of $B_{opt}$.
\begin{align*}
\sup_{\mathcal{D}: \TV(\mathcal{D},U_n)=d>\eps, |B_{opt}(\mathcal{D})|=b}\mathds{E}(\tau_2(\mathcal{D}))&\ge\sup_{\mathcal{D}: \TV(\mathcal{D},U_n)=d>\eps, |B_{opt}(\mathcal{D})|=b} \frac{\log1/3\delta} {\KL(\mathcal{D} , U_n)}.
\\&\ge \frac{\log1/3\delta}{ \KL(b/n\pm\eps/n, b/n)}.
\end{align*}
\end{proof}

\begin{figure}[t]
\centering
\begin{minipage}[c]{.5\linewidth}
\includegraphics[width=\linewidth]{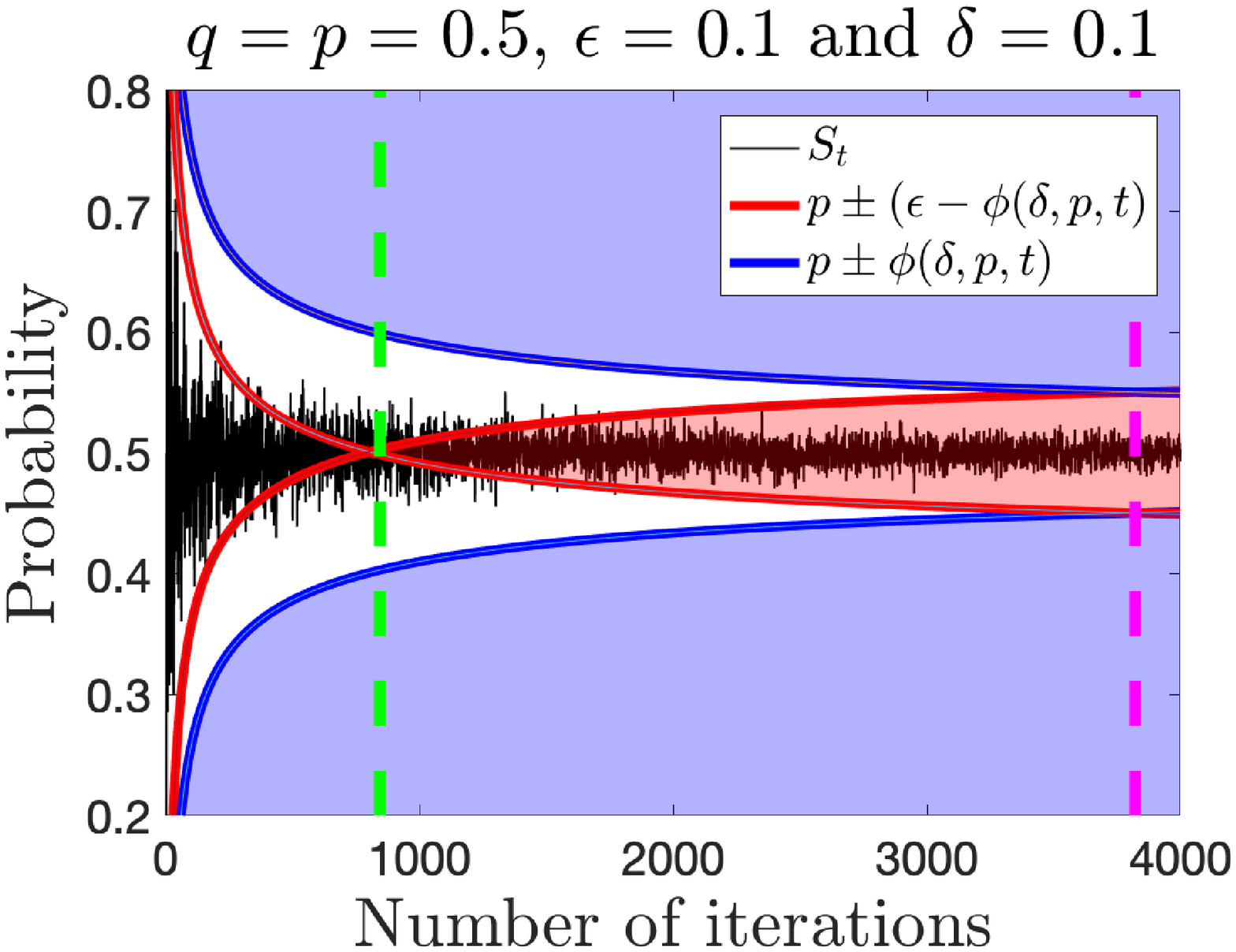}
   \end{minipage}
   \hspace*{-10pt}
   \begin{minipage}[c]{.5\linewidth}
\includegraphics[width=\linewidth]{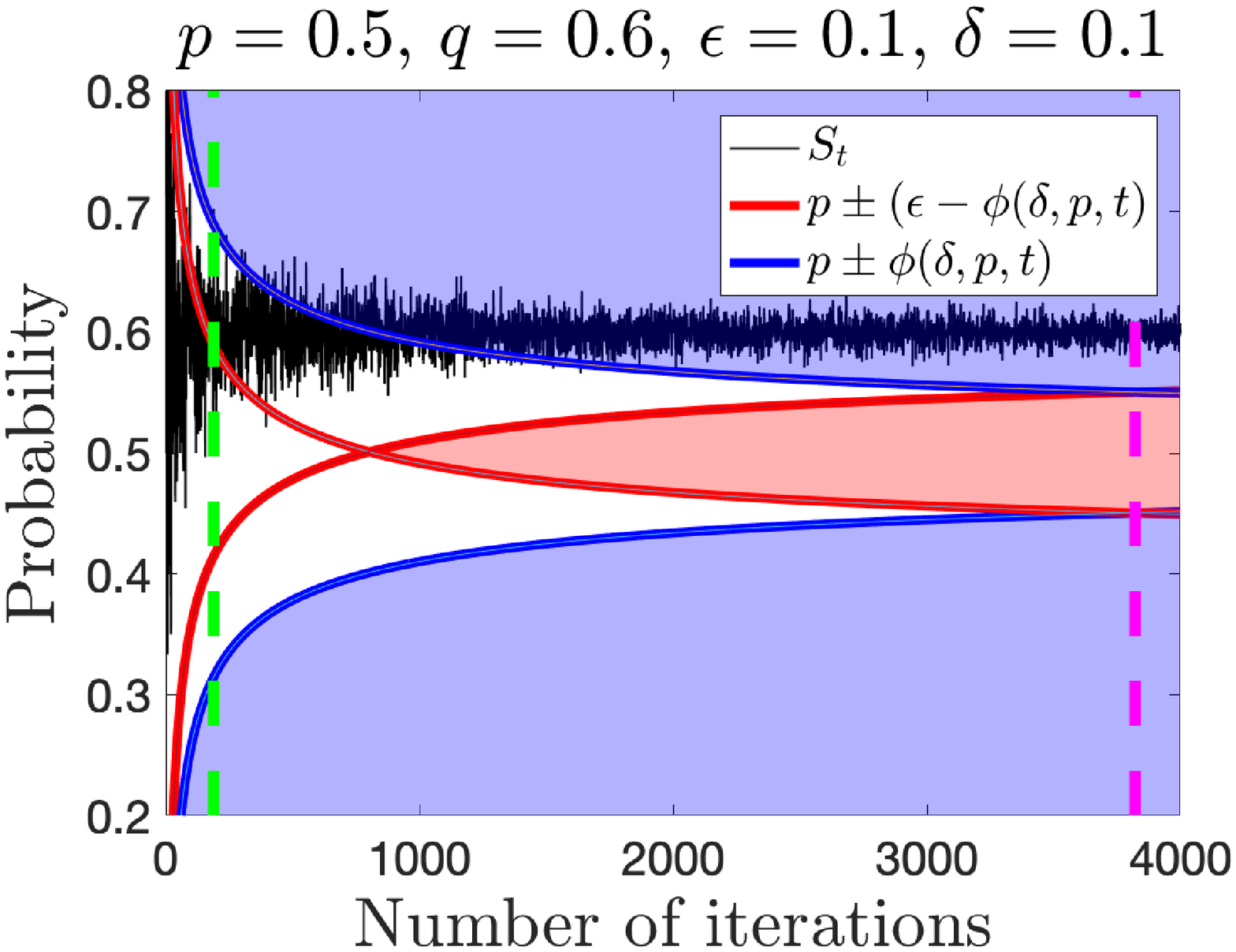}
   \end{minipage}
  \caption{Testing identity for Bernoulli ($n=2$). Left: $q=p=0.5$ and $\varepsilon = 0.1$. Right: $p=0.5$, $q=0.6$ and $\varepsilon=0.1$.
  The sequential tester stops as soon as $S_t$ enters the red region (for $H_1$) or blue region (for $H_2$) whereas the batch tester waits for the red and blue regions to cover the whole segment $[0,1]$. The green and magenta dashed lines represent respectively the stopping time of the sequential and batch algorithms. We note that, in both cases, the sequential tester stops long before the batch algorithm.}
     \label{fig:synthetic2}
\end{figure}
In the sequential testing, the tester chooses when to stop according to the previous observations $(A_1,\dots,A_t)$, making comparisons at each step $t$. The key explanation of the sequential speedup is that the tester can stop as soon as she is sure that she can accept one of the hypothesis $H_1$ or $H_2$. On the contrary, in the batch setting she had to sample enough observation to be simultaneously sure that either $H_1$ or $H_2$ hold. In this aim, at each time step,  after sampling a new observation $X_t$, she compares the updated empirical TV distance  $S_t=\max_{B\subset[\lfloor n/2 \rfloor ]} |\tilde{\mathcal{D}'}_t(B)-|B|/n|$  to specific thresholds  and sees if (a) $S_t$ is sufficiently far from $0$ to surely accept $H_2$, (b) $S_t$ is sufficiently close to $\eps$ to surely accept $H_1$, (c) she is unsure and needs further sample to take a sound decision. This test is formally described in \myalgo{alg} and its execution is illustrated in \myfig{synthetic2} for $n=2$. 

\begin{algorithm}[t]
\caption{Distinguish between $\mathcal{D}'=U_n$ and $\TV(\mathcal{D}',U_n)>\eps$ with high probability}
\label{algo:alg}
\begin{algorithmic}
\REQUIRE $A_1,\dots$ samples from $\mathcal{D}'$  
\ENSURE Accept if $\mathcal{D}'=U_n$ and Reject if $\TV(\mathcal{D}',U_n)>\eps$ with probability of error less than $\delta$
\STATE $t=1$, $W=1$
\WHILE{$W=1$}
\STATE $\tilde{\mathcal{D}}'_{t}=\left\{\left(\sum_{j=1}^t 1_{A_j= i}\right)/t\right\}_{i\in[n]}$
  \IF{$\exists B\subset[\lfloor n/2 \rfloor ]:  |\tilde{\mathcal{D}'}_t(B)-|B|/n|>\max\{\phi(\delta,|B|/n,t),\phi(\delta,1-|B|/n,t)\}$}
    \STATE $W=0$
    \RETURN 2
  \ELSIF{$\forall B\subset[\lfloor n/2 \rfloor ]:   (\tilde{\mathcal{D}'}_t(B)-|B|/n)^+<\eps-\phi(\delta,|B|/n+\eps,t) $ and $ (|B|/n-\tilde{\mathcal{D}'}_t(B))^+<\eps-\phi(\delta,|B|/n-\eps,t)$}
    \STATE $W=0$
    \RETURN 1
  \ELSE
    \STATE $t=t+1$
  \ENDIF
\ENDWHILE
\end{algorithmic}
\end{algorithm}
To control the sample complexity of such sequential algorithms, we need here  Chernoff-Hoeffding's  inequality~\eqref{eq:hoefdding} and the union bound:
\begin{align*}
&\mathds{P}\left(\exists B\subset[\lfloor n/2 \rfloor ], \exists t\ge 1 :  \frac{\sum_{i=1}^{t}1_{A_i\in B}}{t}>\mathcal{D}'(B)+\phi(\delta,\mathcal{D}'(B),t)\right) \\&\le\sum_{B\subset[\lfloor n/2 \rfloor ],t\ge1} e^{-t\KL(\mathcal{D}'(B)+\phi , \mathcal{D}'(B))}=\sum_{B\subset[\lfloor n/2 \rfloor ],t\ge1} \frac{\delta}{2^{n-1}t(t+1)}= \frac{\delta}{2},
\end{align*}
and\begin{align*}
&\mathds{P}\left(\exists B\subset[\lfloor n/2 \rfloor ], \exists t\ge 1 :  \frac{\sum_{i=1}^{t}1_{A_i\in B}}{t}<\mathcal{D}'(B)-\phi(\delta,1-\mathcal{D}'(B),t)\right) 
\\&\le \sum_{B\subset[\lfloor n/2 \rfloor ],t\ge1} e^{-t\KL(\mathcal{D}'(B)-\phi , \mathcal{D}'(B))}=\sum_{B\subset[\lfloor n/2 \rfloor ],t\ge1} \frac{\delta}{2^{n-1}t(t+1)}= \frac{\delta}{2}\;,
\end{align*}
where $\phi(\delta,p,t)$ is the real function\footnote{This function is well defined whenever $\log\left(\frac{2^{n-1}t(t+1)}{\delta}\right)<t\log(1/p)$ and can easily be approximated as a zero of a one-dimensional convex function} implicitly defined as the solution of the equation $\KL(p+\phi(\delta,p,t) , p)=\log\left(\frac{2^{n-1}t(t+1)}{\delta}\right)/t$. The function $\phi$ is a key ingredient in designing the stopping rules of \myalgo{alg} since it enables to directly bound its sample complexity in terms of the expression of the Kullback-Leibler’s divergence.
The stopping time of \myalgo{alg} is $\tau=\min\{\tau_1,\tau_2\}$ where $\tau_1$ and $\tau_2$ are defined as follows:
\begin{align*}
\tau_1=\inf\big\{t\ge 1  : \; &\forall B\subset[\lfloor n/2 \rfloor ]  \; (\tilde{\mathcal{D}}_t(B)-|B|/n)^+<\eps-\phi(\delta,|B|/n+\eps,t),
\\&(|B|/n-\tilde{\mathcal{D}}_t(B))^+<\eps-\phi(\delta,|B|/n-\eps,t)\big\}
\\\tau_2=\inf\big\{t\ge1: \; &\exists B\subset[\lfloor n/2 \rfloor ] \; | \tilde{\mathcal{D}}_t(B)-|B|/n|>\max\{\phi(\delta,|B|/n,t),\phi(\delta,1-|B|/n,t)\}\big\} \ .
\end{align*}
Note that the stopping time of the algorithm is random. Yet, we can show that this algorithm stops before the batch 
 one and give an upper bound on the expected stopping time $\tau$ (or expected sample complexity) 
 using the inequality  $\mathds{E}(\tau)\le N + \sum_{t\ge N}\mathds{P}(\tau\ge t)$, where $N$ is chosen so that $\mathds{P}(\tau\ge t)$ is (exponentially) small for $t\geq N$. 
 In the following theorem, we state an upper bound on the estimated sample complexity of this algorithm. 

\begin{theorem}
The \myalgo{alg} is $\delta$-correct and its stopping times can be bounded in expectation for $n<\log^{2/3}(1/\delta)$ as follows:
\begin{align*}
\mathds{E}(\tau_1(U_n))&\le 
\frac{\log(2^{n-1}/\delta)}{\min_{b\in[n]}\{\KL(b/n,b/n\pm\eps)\}} +\mathcal{O}\left(\frac{\log(2^{n-1}/\delta)^{2/3}}{\eps^2}\right) \text{, and}
\\\mathds{E}(\tau_2(\mathcal{D}'))&\le \frac{\log(2^{n-1}/\delta)}{\min\{\KL(|B_{opt}|/n\pm d,|B_{opt}|/n)\}}+\mathcal{O}\left(\frac{\log(2^{n-1}/\delta)^{2/3}}{d^2}\right) 
\end{align*}
\text{ if } $d=\TV(\mathcal{D}',U_n)=|\mathcal{D}'(B_{opt})-|B_{opt}|/n|>\eps\;$ .
\end{theorem}
These upper bounds are tight in the sense that they match the asymptotic lower bounds of Lemma~\ref{lowerbound-kl-id} if $n \ll\log(1/\delta)$. We see here the many advantages of the sequential setting as shown in Figure \ref{fig:synthetic2}: (a) the sequential algorithm stops always 
before the non-sequential algorithm since after the batch complexity the region of decisions of the sequential algorithm intersect, (b) the estimated sample complexity is $4$ times less than the optimal complexity in the non sequential setting, (c) the sample complexity can be very small if the mass probability is concentrated on a small set i.e. $|B_{opt}|\ll n$ and (d)  the sample complexity in the sequential setting depends on the unknown distribution $\mathcal{D}'$ through the distance $\TV(\mathcal{D}',U_n)$. Note that this cannot be the case in the batch setting as the number of sample should be fixed beforehand. This attribute makes a considerable difference  when $\mathcal{D}'$ is very different from $U_n$. Nevertheless, the above lower bounds and upper bounds do not match exactly, the dependence on $n$ cannot be avoided if $n$ is of the order (or larger) of $\log(1/\delta)$. For this reason, we try in \mysec{test-unif-gen} to somehow truncate our algorithm in a way to get  the best sample complexity in every regime. In the previous results the choice of $\mathcal{D}=U_n$ is not crucial, we can easily generalize them by replacing $|B_{opt}|/n$ by $\mathcal{D}(B_{opt})$. 
\begin{proof}
We first prove the correctness of \myalgo{alg}, i.e., that it has an error probability less than $\delta$. Let us recall the useful concentration inequalities which can be simply proven using  Chernoff-Hoeffding's inequalities and union bounds. 
\begin{lemma}
If $A_1,\dots,A_t$ are i.i.d. random variables with the law $\mathcal{D}'$, we have the following inequalities:
\begin{align*}
&\mathds{P}\left(\exists B\subset[\lfloor n/2 \rfloor ], \exists t\ge 1 :  \frac{\sum_{i=1}^{t}1_{A_i\in B}}{t}>\mathcal{D}'(B)+\phi(\delta,\mathcal{D}'(B),t)\right)  \le\frac{\delta}{2},
\\&\mathds{P}\left(\exists B\subset[\lfloor n/2 \rfloor ], \exists t\ge 1 :  \frac{\sum_{i=1}^{t}1_{A_i\in B}}{t}<\mathcal{D}'(B)-\phi(\delta,1-\mathcal{D}'(B),t)\right) 
\le\frac{\delta}{2}\;,
\end{align*}
\end{lemma}
Using this lemma we can conclude:
\begin{itemize}
\item If $\mathcal{D}'=U_n$, the probability of error is given by \begin{align*}
\mathds{P}\left(\tau_2\le \tau_1\right)&\le  \mathds{P}\left(\exists t \ge 1 : \max_{B\subset[\lfloor n/2 \rfloor ]} |\tilde{\mathcal{D}'}_t(B)-|B|/n|>\max\{\phi(\delta,|B|/n,t),\phi(\delta,1-|B|/n,t)\}\right) 
\\&\le  \mathds{P}\left(\exists t \ge 1, \exists B\subset[\lfloor n/2 \rfloor ]  :  \tilde{\mathcal{D}'}_t(B)-|B|/n>\phi(\delta,|B|/n,t)\right) \\&+\mathds{P}\left(\exists t \ge 1, \exists B\subset[\lfloor n/2 \rfloor ]  :  \tilde{\mathcal{D}'}_t(B)-|B|/n<-\phi(\delta,1-|B|/n,t)\right) 
			\\&\le \delta\;.
\end{align*}
\item If $\mathcal{D}'(B_{opt})-|B_{opt}|/n>\eps$, the probability of error is given by \begin{align*}
\mathds{P}\left(\tau_1\le \tau_2\right)&\le  \mathds{P}\left(\exists t \ge 1 : \tilde{\mathcal{D}}_t'(B_{opt})-|B_{opt}|/n<\eps-\phi(\delta,1-|B_{opt}|-\eps,t)   \right) 
\\&\le \sum_{t\ge 1} e^{-t\KL(1-|B_{opt}|/n-\eps+\phi(\delta,1-|B_{opt}|/n-\eps,t),1-\mathcal{D}'(B_{opt}))}
			\\&\le \sum_{t\ge 1} e^{-t\KL(|B_{opt}|/n+\eps-\phi(\delta,1-|B_{opt}|/n-\eps,t),|B_{opt}|/n+		\eps)}  \text{ as } x\mapsto \KL(p,p-x) \text{ is increasing on } (0,p)
			\\&\le  \sum_{t\ge 1}  \frac{\delta}{2t(t+1)}
			\\&\le \delta\;.
\end{align*}
\item If $\mathcal{D}'(B_{opt})-|B_{opt}|/n<-\eps$, the probability of error is given by \begin{align*}
\mathds{P}\left(\tau_1\le \tau_2\right)&\le  \mathds{P}\left(\exists t \ge 1 :\tilde{\mathcal{D}}_t'(B_{opt})-|B_{opt}|/n >-\eps+\phi(\delta,1-|B_{opt}|/n-\eps,t)   \right) 
\\&\le \sum_{t\ge 1} e^{-t\KL(|B_{opt}|/n-\eps+\phi(\delta,|B_{opt}|/n-\eps,t),\mathcal{D}'(B_{opt}))}
			\\&\le \sum_{t\ge 1} e^{-t\KL(|B_{opt}|/n-\eps+\phi(\delta,|B_{opt}|/n-\eps,t),|B_{opt}|/n-		\eps)} \text{ as } x \mapsto \KL(p,x) \text{ is decreasing on } (0,p)
			\\&\le  \sum_{t\ge 1}  \frac{\delta}{2t(t+1)}
			\\&\le \delta\;.
\end{align*}
\end{itemize}
This conclude the proof of the correctness of \myalgo{alg}. 

Let us prove the upper bounds on the expected stopping times of \myalgo{alg}. We first prove the asymptotic bounds, then provide the proof for the non-asymptotic bounds. 
The proofs rely on the following well-known lemma:
\begin{lemma}\label{lemma-complexity} $T$ a random variable taking values in $\mathds{N}$, we have for all $N\in \mathds{N}^*$ 
\begin{align*}
\mathds{E}(T)\le N+ \sum_{t\ge N}\mathds{P}(T\ge t)\;.
\end{align*}
\end{lemma}

Recall that $\phi$ is defined by the relation $\KL(p+\phi(\delta,p,t),p)=\log\left(\frac{2^{n-1}t(t+1)}{\delta}\right)/t$.
Pinsker's inequality(\cite{reid2009generalised} implies $
0<\phi(\delta,p,t)\le \sqrt{\frac{\log\left(\frac{2^{n-1}t(t+1)}{\delta}\right)}{2t}},
$ hence $\lim_{t\rightarrow\infty} \phi(\delta,p,t)=0$ and we have the equivalent \[\phi(\delta,p,t)\underset{t\rightarrow \infty}{\sim}\sqrt{2p(1-p)\frac{\log\left(\frac{2^{n-1}t(t+1)}{\delta}\right)}{t}}\;.\] Fix a parameter $0<\alpha<1$, and let $N$ the minimum positive integer such that for all integers $t\ge N$, $\max_{b\in[n/2]}\{\phi(\delta,b/n-\eps,t),\phi(\delta,1-b/n-\eps,t)\}\le \alpha \eps.$ The existence of $N$ is guaranteed since $\lim_{t\rightarrow \infty}\phi(\delta,p,t)=0$. We have 
\begin{align*}
\max_{b\in[n/2]}\big\{\phi(\delta,b/n-\eps,N),\phi(\delta,1-b/n-\eps,N)\big\}\le \alpha \eps,
\\\max_{b\in[n/2]}\big\{\phi(\delta,b/n-\eps,N-1),\phi(\delta,1-b/n-\eps,N-1)\big\}> \alpha \eps.
\end{align*}
Hence, 
\begin{align*}
\min_{b\in[n/2]}\big\{\KL(b/n-\eps+\alpha\eps,b/n-\eps),\KL(1-b/n-\eps+\alpha\eps,1-b/n-\eps)\big\} \le \frac{\log(\frac{2^{n-1}(N-1)N}{\delta})}{N-1}.
\end{align*}
Thus $\lim_{\delta\rightarrow 0}N =+\infty$ and from Lemma~\ref{lemma-log} we can deduce 
\begin{align*}
N\le \frac{\log(2^{n-1}/\delta)+4\log\left(\frac{\log(2^{n-1}/\delta)}{\min_{b\in[n/2]}\{\KL(b/n-\eps+\alpha\eps,b/n-\eps),\KL(1-b/n-\eps+\alpha\eps,1-b/n-\eps)\}}\right)}{\min_{b\in[n/2]}\{\KL(b/n-\eps+\alpha\eps,b/n-\eps),\KL(1-b/n-\eps+\alpha\eps,1-b/n-\eps)\}}+1\;.
\end{align*}
Finally, 
\begin{align*}
\limsup_{\delta\rightarrow 0} \frac{N}{\log(1/\delta)}\le \frac{1}{\min_{b\in[n/2]}\big\{\KL(b/n-\eps+\alpha\eps,b/n-\eps),\KL(1-b/n-\eps+\alpha\eps,1-b/n-\eps)\big\}}.
\end{align*}Likewise, for $q=\mathcal{D}'(B_{opt})$ and $p=|B_{opt}|$  such that $|q-p|>\eps$, we can define $N_q$ the minimum positive integer such that for all $t\ge N_q$: 
\begin{align*}
\max\big\{\phi(\delta,p,t),\phi(\delta,1-p,t)\big\}\le \alpha|q-p|\;.
\end{align*}With the same analysis before, we can prove that 
\begin{align*}
N_q\le \frac{\log(2^{n-1}/\delta)+4\log\left(\frac{\log(2^{n_1}/\delta)}{\min\{\KL(p+\alpha|q-p\ , p),\KL(1-p+\alpha|q-p|,1-p)\}}\right)}{\min\{\KL(p+\alpha|q-p|,p),\KL(1-p+\alpha|q-p|,1-p)\}}+1\;.
\end{align*}
Finally, \begin{align*}
\limsup_{\delta\rightarrow 0} \frac{N_q}{\log(1/\delta)}\le \frac{1}{\min\{\KL(p+\alpha|q-p|,p),\KL(1-p+\alpha|q-p|,1-p)\}}\;.
\end{align*}
Then we use Lemma~\ref{lemma-complexity} and make a case study on $\mathcal{D}'$:
\begin{itemize}
\item If $\mathcal{D}'=U_n$, the estimated stopping time can be bound as \begin{align*}
\mathds{E}(\tau_1(U_n))&\le  N+ \sum_{t\ge N}\mathds{P}(\tau_1(U_n)\ge t)
          \\&\le N+ \sum_{s\ge N-1}\mathds{P}(\exists B\subset[\lfloor n/2 \rfloor ]: \Tilde{\mathcal{D}}'_t(B) >|B|/n+\eps-\phi \text{ or } \Tilde{\mathcal{D}}'_t(B)<|B|/n-\eps+\phi)
          \\&\le N+ \sum_{B\subset[\lfloor n/2 \rfloor ], s\ge N-1}\mathds{P}(\Tilde{\mathcal{D}}'_t(B) >|B|/n+\eps-\alpha\eps \text{ or } \Tilde{\mathcal{D}}'_t(B)<|B|/n-\eps+\alpha\eps ) \text{ (by definition of } N)
          \\&\le N+ \sum_{B\subset[\lfloor n/2 \rfloor ],s\ge N-1}\mathds{P}(|\Tilde{\mathcal{D}}'_t(B)- |B|/n|>(1-\alpha)\eps )
          \\&\le N+ \sum_{B\subset[\lfloor n/2 \rfloor ],s\ge N-1}2e^{-2s((1-\alpha)\eps)^2} \text{ (Chernoff-Hoeffding's inequality)} 
                   \\&\le N+ \frac{2^{n/2+1}e^{-2(N-1)((1-\alpha)\eps)^2}}{1-e^{-2((1-\alpha)\eps)^2}}
                   \\&\le N+\frac{2^{n/2+2}e^{-2(N-1)((1-\alpha)\eps)^2}}{(1-\alpha)^2\eps^2}\;,
          \end{align*} where we use the inequality $1-e^{-x} \ge x/2$ for $0<x<1$ in the last line.
\item If $q:=\mathcal{D}'(B_{opt})>|B_{opt}|/n+\eps=:p+\eps$, the estimated stopping time can be bound as \begin{align*}\mathds{E}(\tau_2)&\le N_q+ \sum_{t\ge N}\mathds{P}(\tau_2(\mathcal{D}')\ge t)
          \\&\le N_q+ \sum_{s\ge N-1}\mathds{P}(|\tilde{\mathcal{D}}'_s(B_{opt})-p|\le \max\{\phi(\delta,p,t),\phi(\delta,1-p,t)\})
          \\&\le N_q+ \sum_{s\ge N-1}\mathds{P}(\tilde{\mathcal{D}}'_s(B_{opt}) \le p+\max\{\phi(\delta,p,t),\phi(\delta,1-p,t)\})
          \\&\le N_q+ \sum_{s\ge N-1}\mathds{P}(\tilde{\mathcal{D}}'_s(B_{opt}) \le p+\alpha(q-p))  \text{ (by definition of } N_q)
          \\&\le N_q+ \sum_{s\ge N-1}\mathds{P}(\tilde{\mathcal{D}}'_s(B_{opt}) \le q-(1-\alpha)(q-p))
          \\&\le N_q+ \sum_{s\ge N-1}e^{-2s((1-\alpha)(q-p))^2} \text{ (Chernoff-Hoeffding's inequality)} 
                   \\&\le N_q+ \frac{e^{-2(N-1)((1-\alpha)(q-p))^2}}{1-e^{-2((1-\alpha)(q-p))^2}}
                   \\&\le N_q+\frac{1}{(1-\alpha)^2(q-p)^2}\;.
          \end{align*} 
\item If $q:=\mathcal{D}'(B_{opt})<|B_{opt}|/n-\eps=:p-\eps$, the estimated stopping time can be bound as \begin{align*}\mathds{E}(\tau_2(\mathcal{D}'))
&\le N_q+ \sum_{t\ge N}\mathds{P}(\tau_2(\mathcal{D}')\ge t)
          \\&\le N_q+ \sum_{s\ge N-1}\mathds{P}(|\tilde{\mathcal{D}}'_s(B_{opt})-p|\le \max\{\phi(\delta,p,t),\phi(\delta,1-p,t)\})
          \\&\le N_q+ \sum_{s\ge N-1}\mathds{P}(\tilde{\mathcal{D}}'_s(B_{opt}) \ge p-\max\{\phi(\delta,p,t),\phi(\delta,1-p,t)\})
          \\&\le N_q+ \sum_{s\ge N-1}\mathds{P}(\tilde{\mathcal{D}}'_s(B_{opt}) \ge p-\alpha(p-q))  \text{ (by definition of } N_q)
          \\&\le N_q+ \sum_{s\ge N-1}\mathds{P}(\tilde{\mathcal{D}}'_s(B_{opt}) \ge q+(1-\alpha)(p-q))
          \\&\le N_q+ \sum_{s\ge N-1}e^{-2s((1-\alpha)(q-p))^2} \text{ (Chernoff-Hoeffding's inequality)} 
                   \\&\le N_q+ \frac{e^{-2(N-1)((1-\alpha)(q-p))^2}}{1-e^{-2((1-\alpha)(q-p))^2}}
                   \\&\le N_q+\frac{1}{(1-\alpha)^2(q-p)^2}\;.
          \end{align*} 
\end{itemize}
Dividing by $\log(1/\delta)$, taking the limits $\delta\rightarrow 0$ then $\alpha\rightarrow 1$ permit to deduce the asymptotic bounds.
By choosing $\alpha=(1+\log(2^{n-1}/\delta)^{-1/3})^{-1}$, we conclude for $n<\log^{2/3}(1/\delta)$:
\begin{align*}
\mathds{E}(\tau_1(U_n))&\le 
\frac{\log(2^{n-1}/\delta)}{\min_{b\in[n]}\{\KL(b/n,b/n\pm\eps)\}} +\mathcal{O}\left(\frac{\log(2^{n-1}/\delta)^{2/3}}{\eps^2}\right) \text{, and}
\\\mathds{E}(\tau_2(\mathcal{D}'))&\le \frac{\log(2^{n-1}/\delta)}{\min\{\KL(|B_{opt}|/n\pm d,|B_{opt}|/n)\}}+\mathcal{O}\left(\frac{\log(2^{n-1}/\delta)^{2/3}}{d^2}\right) 
\end{align*}
\text{ if } $d=\TV(\mathcal{D}',U_n)=|\mathcal{D}'(B_{opt})-|B_{opt}|/n|>\eps\;$. 
\end{proof}

\section{Testing closeness for small $n$\label{sec:bernoulli}}

In this section, $n\ge2$ is still small and we consider this time  two unknown distributions $\mathcal{D}_1$ and $\mathcal{D}_2$ on $[n]$.  We are testing two hypothesis $H_1$: $\mathcal{D}_1=\mathcal{D}_2$ and $H_2$: $\TV(\mathcal{D}_1,\mathcal{D}_2)>\eps$. Similar to the previous section, we are interested in precisely comparing the sample complexity of testing closeness in the sequential versus the batch setting. 
In order to find the optimal constant, we first need to obtain a sharp lower bound in the batch setting, which is done directly by using Stirling's approximation. We then turn to the sequential case.

\subsection{Batch setting}\label{sec:batch}


In the batch setting, the number of steps $\tau $ is fixed before the test. The tester samples $A_1,\dots,A_\tau \sim \mathcal{D}_1$ and $B_1,\dots,B_\tau \sim \mathcal{D}_2$ then decides according to the comparison between the empirical TV distance  $\TV(\tilde{\mathcal{D}_1}_\tau,\tilde{\mathcal{D}_2}_\tau)$ and $\eps/2$ where $\tilde{\mathcal{D}_1}_{\tau}=\left\{\left(\sum_{j=1}^{\tau} 1_{A_j= i}\right)/\tau\right\}_{i\in[n]}$ and $\tilde{\mathcal{D}_2}_{\tau}=\left\{\left(\sum_{j=1}^{\tau} 1_{B_j= i}\right)/\tau\right\}_{i\in[n]}$ are the empirical distributions. If  $\TV(\tilde{\mathcal{D}_1}_\tau,\tilde{\mathcal{D}_2}_\tau)\le \eps/2$,  she accepts $H_1$  and rejects it otherwise. In order to control the number of steps $\tau$ so that the error of this algorithm does not exceed $\delta$, McDiarmid's inequality (\cite{habib2013probabilistic}) writes for $\tau=\frac{4\log(2^{\lfloor n/2 \rfloor }/\delta)}{\eps^2} $:
\begin{align}\label{eq:M}\tag{M}
\mathds{P}\Bigg(\exists B\subset [\lfloor n/2\rfloor] :\left|\tilde{\cD}_{1,\tau}(B)-\cD_1(B)-\tilde{\cD}_{2,\tau}(B)+\cD_2(B)\right|>\frac{\eps}{2}\Bigg)\le \sum_{B\subset [\lfloor n/2\rfloor]}e^{-\tau\eps^2/4}\le \delta\;.
\end{align}
Using 
the concentration inequality~\eqref{eq:M} for $\cD_1=\cD_2$ (to control the type I error) 
 and  for $\mathcal{D}_1\neq \cD_2$ (to control  the type II error)
we prove that this test is $\delta$-correct.
We show in the following theorem that this number of steps $\tau$ is necessary. 
\begin{proposition}\label{ber-clos-ns}
In the batch setting, the algorithm consisting of accepting $H_1$ when $\TV(\tilde{\mathcal{D}_1}_\tau,\tilde{\mathcal{D}_2}_\tau)\le \eps/2$ and rejecting it otherwise is $\delta$-correct for $\tau=\frac{4\log(2^{\lfloor n/2\rfloor}/\delta)}{\eps^2} $.

Moreover, any $\delta$-correct algorithm testing closeness requires at least $\tau$ samples, where
\begin{align*}
\tau\ge \min\bigg\{\frac{\log(1/2\delta)}{2\KL(1/2-\eps/4,1/2-\eps/2)},\frac{\log(1/2\delta)}{2\KL(1/2+\eps/4,1/2)}\bigg\}-\mathcal{O}\left(\frac{\log\log(1/\delta)}{\eps^2}\right) \;.
\end{align*}

\end{proposition}
For this proof, we show that every $\delta$-correct tester can be transformed into a test which depends only on the numbers of $1's, 2's,\dots,n's$ occurred on $\{A_1,\dots,A_\tau\}$ and $\{B_1,\dots,B_\tau\}$. We then consider the distributions $\mathcal{D}_{1,2}=\{1/2,1/2,0,\dots,0\}$ or $\mathcal{D}_{1,2}=\{1/2\pm \eps/2,1/2 \mp \eps/2,0,\dots,0\}$ depending on the outcome of the algorithm when it sees two words of samples having respectively $\tau(1/2-\eps/4)$ and $\tau(1/2+\eps/4)$ ones (the rest of samples are equal to $2$) and derive tight lower bounds on the probability mass function of the multinomial distribution.
\begin{proof}
We consider distributions supported only on $\{1,2\}$, this is possible since we want that our algorithm would work for all distributions. 
We consider such a $\delta$-correct test $A:\{1,2\}^\tau\times \{1,2\}^\tau \rightarrow \{1,2\}$, it sees two words consisting of $\tau$ samples either from equal distributions or $\eps$-far ones and returns $1$ if it thinks they are equal and $2$ otherwise. We construct another test $B:\{1,2\}^\tau\times \{1,2\}^\tau \rightarrow \{0,1\}$ by the expression 
\begin{align*}
B(x,y)=1_{\sum_{\sigma, \rho\in \mathcal{S}_\tau}A(\sigma(x),\rho(y))-1 \ge (\tau!)^2/2 } \;,
\end{align*}
$B$ can be proven to be $2\delta$-correct and have the property of invariance under the action of the symmetric group.  
This leads to an algorithm $C:\{0,\dots,\tau\}^2\rightarrow \{0,1\}$ which is  $2\delta$ correct and satisfies 
\begin{align*}
C(i,j)=B(x_i,y_j)\;,
\end{align*}where $x_k=1\dots 1 2\dots 2$ with $k$ ones. We consider $i=[\tau(1/2-\eps/4)]$ and $j=[\tau(1/2+\eps/4)]$. We denote by $N_i(x)$ the number of $i$ in a word $x$ of length $\tau$  for $i=1,2$.
\begin{itemize}
\item If $C(i,j)=0$, let $x$ (resp. $y$) a word of length $\tau$ constituted of i.i.d samples from $\{1/2-\eps/2,1/2+\eps/2,0,\dots,0\}$ (resp. $\{1/2+\eps/2,1/2-\eps/2,0,\dots,0\}$),  then $\mathds{P}_{1/2-\eps/2,1/2+\eps/2}(N_1(x)=i,N_1(y)=j)\le 2\delta$ hence with Stirling's approximation (\citet{leubner1985generalised} )\begin{align*}
\frac{e^{-2}}{2\pi \tau}e^{-\tau\KL(i/\tau,1/2-\eps/2)}e^{-\tau\KL(1-j/\tau,1/2-\eps/2)}\le 2\delta.
\end{align*}Thus 
\begin{align*}
2\tau\KL(1/2+\eps/4-1/\tau,1/2+\eps/2)&\ge \tau(\KL(i/\tau,1/2-\eps/2)+\KL(j/\tau,1/2-\eps/2))\\&\ge \log(1/2\delta)-2-\log(2\pi)-\log(\tau)\;.
\end{align*}
Hence using lemma~\ref{lemma-kl} and for $\tau>2/\eps$ 
\begin{align*}
2\tau\KL(1/2+\eps/4,1/2+\eps/2)&\ge -2\tau(\KL(1/2+\eps/4-1/\tau,1/2+\eps/2)-\KL(1/2+\eps/4,1/2+\eps/2) )\\&+ \log(1/2\delta)-2-\log(2\pi)-\log(\tau)
\\&\ge -2\tau\int_{1/2+\eps/4-1/\tau}^{1/2+\eps/4}du \int^{1/2+\eps/2}_u dv\frac{1}{v(1-v)} +\log(1/2\delta)\\&-2-\log(2\pi)-\log(\tau)
\\&\ge -2(\eps/4+1/\tau)\sup_{[1/2+\eps/4-1/\tau,1/2+\eps/2]}\frac{1}{v(1-v)} +\log(1/2\delta)
			\\&-2-\log(2\pi)-\log(\tau)
\\&\ge -2\eps\sup_{[1/2,1/2+\eps]}\frac{1}{v(1-v)} +\log(1/2\delta)-2-\log(2\pi)-\log(\tau)\;.
\end{align*}
Then lemma \ref{lemma-log} implies 
\begin{align*}
\tau &\ge \frac{-2\eps\sup_{[1/2,1/2+\eps]}\frac{1}{v(1-v)} +\log(1/2\delta)-2-\log(2\pi)}{2\KL(1/2+\eps/4,1/2+\eps/2)}-\frac{\log\left(\frac{-2\eps\sup_{[1/2,1/2+\eps]}\frac{1}{v(1-v)} +\log(1/2\delta)-2-\log(2\pi)}{2\KL(1/2+\eps/4,1/2+\eps/2)}\right)}{4\KL(1/2+\eps/4,1/2+\eps/2)}\;
          \\&\ge \frac{\log(1/2\delta)}{2\KL(1/2+\eps/4,1/2+\eps/2)}-\mathcal{O}\left( \frac{\log\log(1/\delta)}{\KL(1/2+\eps/4,1/2+\eps/2)}\right).
\end{align*}

Finally we get the asymptotic lower bound: \begin{align*}
\liminf_{\delta\rightarrow 0} \frac{\tau}{\log(1/\delta)}\ge\frac{1}{2\KL(1/2-\eps/4,1/2-\eps/2)}.
\end{align*}
\item If $C(i,j)=1$, let $x$ and $y$ two words of length $\tau$ constituted of i.i.d samples from $\{1/2,1/2,0,\dots,0\}$, then $\mathds{P}_{1/2,1/2}(N_1(x)=i,N_1(y)=j)\le 2\delta$ hence with Stirling's approximation \begin{align*}
\frac{e^{-2}}{2\pi \tau}e^{-\tau\KL(i/\tau,1/2)}e^{-\tau\KL(1-j/\tau,1/2)}\le 2\delta\;.
\end{align*}Using the same lemmas as before, we get the following lower bound \begin{align*}
\tau &\ge \frac{\log(1/2\delta)}{2\KL(1/2+\eps/4,1/2)}-\mathcal{O}\left( \frac{\log\log(1/\delta)}{\KL(1/2+\eps/4,1/2)}\right).
\end{align*}

Finally we get the asymptotic lower bound: \begin{align*}
\liminf_{\delta\rightarrow 0} \frac{\tau}{\log(1/\delta)}\ge\frac{1}{2\KL(1/2+\eps/4,1/2)}\;.
\end{align*}
\end{itemize}
\end{proof}

\subsection{Sequential setting}



Inspired by testing identity for small alphabets results in \mysec{bernoulli-id}, we would like to  achieve an improvement of factor  $4$ in sample complexity for sequential strategies over batch ones for testing closeness problem. For this end, 
we start by stating a lower bound on the expected stopping times of a sequential algorithm for testing closeness.

\begin{lemma}\label{lowerbound-kl-clos}
 Let $T$ be a stopping rule for testing closeness: $\mathcal{D}_1=\cD_2$ vs ${\TV(\mathcal{D}_1,\cD_2)>\eps}$ with an error probability $\delta$. Let  $\tau_1$ and $\tau_2$ the associated stopping times. We have \begin{align*}
\sup_{\cD}\mathds{E}(\tau_1(\cD,\cD))&\ge	\frac{\log(1/3\delta)}{\KL(1/2,1/2+\eps/2)+\KL(1/2,1/2-\eps/2)} \underset{\eps\rightarrow 0}\sim \frac{\log(1/3\delta)}{\eps^2}	 \text{ and }
\\ \sup_{\TV(\cD_1,\cD_2)=d} \mathds{E}(\tau_2(\cD_1,\cD_2))&\ge  \frac{\log(1/3\delta)}{\KL(1/2+d/2,1/2)+\KL(1/2-d/2,1/2)} \underset{d\rightarrow 0}\sim \frac{\log(1/3\delta)}{d^2}     \text{ if } d>\eps  \;.
\end{align*}
\end{lemma}
An average number of samples equivalent to $\log(1/3\delta)(\eps\vee \TV(\cD_1,\cD_2))^{-2}  $ (see Lemma \ref{lemma-kl}) is thus necessary when the tester can access sequentially to the samples, which is roughly 4 times less than the complexity  obtained in the batch setting. 
\begin{proof}
The proof of this Lemma follows from Lemma~\ref{lowerbound-kl} by choosing for the first point $\cD_1=\cD_2=\{1/2,1/2,0,\dots,0\}$ and $\cD'_{1,2}=\{1/2\pm \eps/2,1/2 \mp \eps/2,0,\dots,0\}$. For the second point, we use $\cD=\{1/2,1/2,0,\dots,0\}$ and $\cD_{1,2}=\{1/2\pm d/2,1/2 \mp d/2,0,\dots,0\}$.
\end{proof}

In the sequential testing, the tester chooses when to stop according to the previous observations $((A_1,B_1),\dots,(A_t,B_t))$, making comparisons at each step $t$. 
The tester can stop as soon as she is sure that she can accept one of the hypothesis $H_1$ or $H_2$. 
On the contrary, in the batch setting she had to sample enough observation to be simultaneously sure that either $H_1$ or $H_2$ hold. In this aim, at each time step,  after sampling a new observation $(A_t,B_t)$, she compares the updated empirical TV distance  $S_t=\TV(\tilde{\mathcal{D}_1}_t,\tilde{\mathcal{D}_2}_t)$  to specific thresholds  and sees if (a) $S_t$ is sufficiently far from $0$ to surely accept $H_2$, (b) $S_t$ is sufficiently close to $\eps$ to surely accept $H_1$, (c) she is unsure and needs further samples to take a sound decision. This test is formally described in \myalgo{alg-clos} and its execution is illustrated in \myfig{synthetic} for $n=2$. 

\begin{algorithm}[t]
	\caption{Distinguish between $\cD_1=\cD_2$ and $\TV(\cD_1,\cD_2)>\eps$ with high probability}
	\label{algo:alg-clos}
	\begin{algorithmic}
		\REQUIRE $A_1,\dots$ samples from $\mathcal{D}_1$ and $B_1,\dots$ samples from $\mathcal{D}_2$ 
		\ENSURE Accept if $\cD_1=\cD_2$ and Reject if $\TV(\cD_1,\cD_2)>\eps$ with probability of error less than $\delta$
		\STATE $t=1$, $W=1$
	 \WHILE{$W=1$}
		\STATE $\tilde{\mathcal{D}}_{1,t}=\left\{\left(\sum_{j=1}^{t} 1_{A_j= i}\right)/t\right\}_{i\in[n]}$, $\tilde{\mathcal{D}}_{2,t}=\left\{\left(\sum_{j=1}^{t} 1_{B_j= i}\right)/t\right\}_{i\in[n]}$
		\IF{$\TV\left(\tilde{\mathcal{D}}_{1,t},\tilde{\mathcal{D}}_{2,t}\right)>\sqrt{\frac{\log\left(\frac{2^{n-1}t(t+1)}{\delta}\right)}{t}}$}
		\STATE $W=0$
		\RETURN 2
		\ELSIF{$\TV\left(\tilde{\mathcal{D}}_{1,t},\tilde{\mathcal{D}}_{2,t}\right)\le \eps-\sqrt{\frac{\log\left(\frac{2^{n-1}t(t+1)}{\delta}\right)}{t}}$}
		\STATE $W=0$
		\RETURN 1
		\ELSE
		\STATE $t=t+1$
		\ENDIF
		\ENDWHILE
	\end{algorithmic}
\end{algorithm}
To show the correctness of such sequential algorithms, Chernoff-Hoeffding's inequality doesn't work since we have two unknwon distributions. It turns out that McDiarmid's inequality~\eqref{eq:M} is best suited in this situation: 
\begin{align*}
	\mathds{P}\Bigg(\exists t\ge 1 , \exists B\subset [n/2] :\left|\tilde{\cD}_{1,t}(B)-\cD_1(B)-\tilde{\cD}_{2,t}(B)+\cD_2(B)\right|>\Phi_t\Bigg)\le \delta, 
\end{align*}
where $\Phi_t$  denote the constant $ \Phi_t=\sqrt{\log\left(\frac{2^{n-1}t(t+1)}{\delta}\right)/t}$.
On the other hand, to control the sample complexity, we prove upper bounds on the expected stopping times:
\begin{align*}
 \tau_1=\inf\left\{ t\ge 1 :     \TV\left(\tilde{\mathcal{D}}_{1,t},\tilde{\mathcal{D}}_{2,t}\right)\le\eps-\Phi_t\right\}, \text{ and } 
 \tau_2=\inf\left\{ t\ge 1 :     \TV\left(\tilde{\mathcal{D}}_{1,t},\tilde{\mathcal{D}}_{2,t}\right)>\Phi_t \right\}.
 \end{align*}
It is clear that the stopping time of the algorithm is random. Yet, we can show that  this algorithm stops before the non sequential 
 one and give an upper bound on the expected stopping time $\tau$ (or expected sample complexity). 
 In the following theorem, we state an upper bound on the estimated sample complexity of this algorithm. 
 \begin{theorem}\label{ber-clos-comp2}
 	The \myalgo{alg-clos} is $\delta$-correct and its stopping times verify for $n\le\mathcal{O}( \log(1/\delta)^{1/3})$:
 	\begin{align*}
 	\mathds{E}(\tau_1(\cD,\cD)) &\le 
 	\frac{\log(2^{n+1}/\delta)}{\eps^2}+\mathcal{O}\left( \frac{\log(2^{n+1}/\delta)^{2/3}}{\eps^2}\right)  \text{ if } \cD_1=\cD_2=\cD \text{ and}
 	\\\mathds{E}(\tau_2(\cD_1,\cD_2)) &\le 
 	\frac{\log(2^{n+1}/\delta)}{\TV(\cD_1,\cD_2)^2}+\mathcal{O}\left( \frac{\log(2^{n+1}/\delta)^{2/3}}{\TV(\cD_1,\cD_2)^2}\right) \text{ if } \TV(\cD_1,\cD_2)>\eps\;.
 	\end{align*}
  \end{theorem}
These upper bounds are tight in the sense that they match the asymptotic lower bounds of Lemma~\ref{lowerbound-kl-clos} if $n \ll\log(1/\delta)$. The advantages of sequential strategies over batch ones for the testing closeness problem are the same as those for the testing identity problem.
\begin{proof}
 	We should prove that the \myalgo{alg-clos} has an error probability less than $\delta$. We use the following lemma which can be proven using McDiarmid's inequality and union bounds. 
 	\begin{lemma}
 		If $\{A_1,\dots,A_{t}\}$ (resp $\{B_1,\dots,B_{t}\}$ ) i.i.d. with the law $\mathcal{D}_1$ (resp $\mathcal{D}_2$), we have the following inequality 
 		\[
 		\mathds{P}\Bigg(\exists t\ge 1 , \exists B\subset [n/2] :\left|\tilde{\cD}_{1,t}(B)-\cD_1(B)-\tilde{\cD}_{2,t}(B)+\cD_2(B)\right|>\sqrt{\log\left(\frac{2^{n-1}t(t+1)}{\delta}\right)/t}\Bigg)\le \delta. 
 		\]
 	\end{lemma}
 Using this lemma we can conclude:
 \begin{itemize}
 	\item If $\cD_1=\cD_2$, the probability of error is given by \[
 	\mathds{P}\left(\tau_2\le \tau_1\right)\le  \mathds{P}\left(\exists t \ge 1 : \TV\left(\tilde{\mathcal{D}}_{1,t},\tilde{\mathcal{D}}_{2,t}\right)>\sqrt{\log\left(\frac{2^{n-1}t(t+1)}{\delta}\right)/t}\right) 
 	\le \delta\;.\]
 	\item If $\TV(\cD_1,\cD_2)=|\cD_1(B_{opt})-\cD_2(B_{opt})|>\eps$, the probability of error is given by \begin{align*}
 	\mathds{P}\left(\tau_1\le \tau_2\right)&\le\mathds{P}\left(\exists t \ge 1 : \TV\left(\tilde{\mathcal{D}}_{1,t},\tilde{\mathcal{D}}_{2,t}\right)\le\eps-\sqrt{\log\left(\frac{2^{n-1}t(t+1)}{\delta}\right)/t}\right) 
  \\&\le \mathds{P}\left(\exists t\ge 1 :  \left|\tilde{\cD}_{1,t}(B_{opt})-\tilde{\cD}_{2,t}(B_{opt}))\right|\le\eps-\sqrt{\log\left(\frac{2^{n-1}t(t+1)}{\delta}\right)/t}\right)
 	\\&\le \mathds{P}\Bigg(\exists t\ge 1 :  \left|\tilde{\cD}_{1,t}(B_{opt})-\cD_1(B_{opt})-\tilde{\cD}_{2,t}(B_{opt})+\cD_2(B_{opt}))\right|\ge|\cD_1(B_{opt})-\cD_2(B_{opt})|
 	\\&-\eps+\sqrt{\log\left(\frac{2^{n-1}t(t+1)}{\delta}\right)/t}\Bigg)
 	\\&\le \mathds{P}\Bigg(\exists t\ge 1 :  \left|\tilde{\cD}_{1,t}(B_{opt})-\cD_1(B_{opt})-\tilde{\cD}_{2,t}(B_{opt})+\cD_2(B_{opt}))\right|
 	\\&>\sqrt{\log\left(\frac{2^{n-1}t(t+1)}{\delta}\right)/t}\Bigg)
 	\\&\le \delta\;.
 	\end{align*}
 \end{itemize}
These computations prove the correctness of \myalgo{alg-clos}.
It remains to study  the complexity of  \myalgo{alg-clos}. To this aim, we make a case study and use lemma \ref{lemma-complexity} to upper bound the stopping rules. 

Let us take $\alpha\in (0,1)$, 
\begin{itemize}
	\item If $\cD_1=\cD_2$, we take $N=\left[\frac{\log(2^{n+1}/\delta)}{(\alpha\eps)^2}\right]+1$ and $\tilde{\alpha}\in (0,1)$\footnote{for fixed $\alpha$ we take $\delta$ small enough to have $\tilde{\alpha} <1$.} so that 
	\[
	\tilde{\alpha}^2=\alpha^2\left(\frac{\log\log(2^{n+1}/\delta)-\log((\alpha\eps)^2)}{\log(2^{n+1}/\delta)}+1\right).\]
	The estimated stopping time can be bound as \begin{align*}
	\mathds{E}(\tau_1(\cD_1,\cD_2))&\le N+ \sum_{s\ge N}\mathds{P}(\tau_1(\cD_1,\cD_2)\ge s)
	\\&\le N+ \sum_{t\ge N-1}\mathds{P}\left( \TV\left(\tilde{\mathcal{D}}_{1,t},\tilde{\mathcal{D}}_{2,t}\right)>\eps-\sqrt{\log\left(\frac{2^{n-1}t(t+1)}{\delta}\right)/t} \right)
	\\&\le N+ \sum_{t\ge N-1}\mathds{P}\left( \TV\left(\tilde{\mathcal{D}}_{1,t},\tilde{\mathcal{D}}_{2,t}\right)>\eps-\tilde{\alpha}\eps\right)
	\\&\le N+ \sum_{t\ge N-1}\mathds{P}\left( \TV\left(\tilde{\mathcal{D}}_{1,t},\tilde{\mathcal{D}}_{2,t}\right)>(1-\tilde{\alpha})\eps\right)
	\\&\le N+ \sum_{t\ge N-1}2^{n/2}e^{-t((1-\tilde{\alpha})\eps)^2}, (\text{McDiarmid's inequality})
	\\&\le N+ \frac{2^{n/2}e^{-(N-1)((1-\tilde{\alpha})\eps)^2}}{1-e^{-((1-\tilde{\alpha})\eps)^2}}
	\\&\le \frac{\log(2^{n+1}/\delta)}{(\alpha\eps)^2} +2\frac{2^{n/2}e^{-(N-1)((1-\tilde{\alpha})\eps)^2}}{((1-\tilde{\alpha})\eps)^2}+1\;, (1-e^{-x} \ge x/2 \text{ for } 0<x<1)
	\\&\le \frac{\log(2^{n+1}/\delta)}{\eps^2} +\frac{\log(2^{n+1}/\delta)^{2/3}}{\eps^2} +\mathcal{O}\left(\frac{\log(2^{n+1}/\delta)^{2/3}}{\eps^2}\right)\;
		\\&\le \frac{\log(2^{n+1}/\delta)}{\eps^2}  +\mathcal{O}\left(\frac{\log(2^{n+1}/\delta)^{2/3}}{\eps^2}\right)\;,
	\end{align*} for $\alpha=(1+\log(2^{n+1}/\delta)^{-1/3})^{-2}$ so that $1-\tilde{\alpha}\ge C\log(2^{n+1}/\delta)^{-1/3} $ and we suppose here that $n<2C^2\log(2^{n+1}/\delta)^{1/3}$.
	\item If $d=\TV(\cD_1,\cD_2)=|\cD_1(B_{opt})-\cD_2(B_{opt})|>\eps$, we take $N=\left[\frac{\log(2^{n+1}/\delta)}{(\alpha d)^2}\right]+1$. We take $\tilde{\alpha}\in (0,1)$ so that $\tilde{\alpha}^2=\alpha^2\left(\frac{\log\log(2^{n+1}/\delta)-\log((\alpha d)^2)}{\log(2^{n+1}/\delta)}+1\right)$. The estimated stopping time can be bound as \begin{align*}
	\mathds{E}(\tau_2(\cD_1,\cD_2))&\le N+ \sum_{s\ge N}\mathds{P}(\tau_2(\cD_1,\cD_2)\ge s)
	\\&\le N+ \sum_{t\ge N-1}\mathds{P}\left( \TV\left(\tilde{\mathcal{D}}_{1,t},\tilde{\mathcal{D}}_{2,t}\right)
	\le\sqrt{\log\left(\frac{2^{n-1}t(t+1)}{\delta}\right)/t} \right)
		\\&\le N+ \sum_{t\ge N-1}\mathds{P}\left( \TV\left(\tilde{\mathcal{D}}_{1,t},\tilde{\mathcal{D}}_{2,t}\right)
	\le\sqrt{\log\left(\frac{2^{n-1}t(t+1)}{\delta}\right)/t} \right)
	 \\&\le N+ \sum_{t\ge N-1} \mathds{P}\left(  \left|\tilde{\cD}_{1,t}(B_{opt})-\tilde{\cD}_{2,t}(B_{opt}))\right|\le\sqrt{\log\left(\frac{2^{n-1}t(t+1)}{\delta}\right)/t}\right)
	\\& \le N+ \sum_{t\ge N-1}\mathds{P}\Bigg(  \left|\tilde{\cD}_{1,t}(B_{opt})-\cD_1(B_{opt})-\tilde{\cD}_{2,t}(B_{opt})+\cD_2(B_{opt}))\right|
	\\&>|\cD_1(B_{opt})-\cD_2(B_{opt})|-\sqrt{\log\left(\frac{2^{n-1}t(t+1)}{\delta}\right)/t}\Bigg)
	\\& \le N+ \sum_{t\ge N-1} \mathds{P}\Bigg(  \left|\tilde{\cD}_{1,t}(B_{opt})-\cD_1(B_{opt})-\tilde{\cD}_{2,t}(B_{opt})+\cD_2(B_{opt}))\right|>(1-\tilde{\alpha})d \Bigg)
	\\&\le N+ \sum_{t\ge N-1}e^{-t((1-\tilde{\alpha})d)^2}
	\\&\le N+ \frac{e^{-(N-1)((1-\tilde{\alpha})d)^2}}{1-e^{-((1-\tilde{\alpha})d)^2}}
	\\&\le \frac{\log(2^{n+1}/\delta)}{(\alpha d)^2} +\frac{2}{(1-\tilde{\alpha})^2d^2}+1\;
		\\&\le \frac{\log(2^{n+1}/\delta)}{d^2}  +\mathcal{O}\left(\frac{\log(2^{n+1}/\delta)^{2/3}}{d^2}\right)\;, 
	\end{align*} where we choose  $\alpha=(1+\log(2^{n+1}/\delta)^{-1/3})^{-2}$ and we use the inequality $1-e^{-x} \ge x/2$ for $0<x<1$ in the last line.
\end{itemize}
Finally, we can deduce the limit when $\cD_1=\cD_2$:
\begin{align*}
\limsup_{\delta\rightarrow 0} \frac{\mathds{E}(\tau_1(\cD_1,\cD_2))}{\log(1/\delta)} &\le \limsup_{\delta\rightarrow 0} \frac{\log(2^{n+1}/\delta)}{\log(1/\delta)\eps^2} +\mathcal{O}\left(\frac{\log(2^{n+1}/\delta)^{2/3}}{\log(1/\delta)\eps^2}\right)
\\&\le \frac{1}{\eps^2}\;,
\end{align*}
and when $d=\TV(\cD_1,\cD_2)>\eps$:

\begin{align*}
\limsup_{\delta\rightarrow 0} \frac{\mathds{E}(\tau_2(\cD_1,\cD_2))}{\log(1/\delta)} &\le \limsup_{\delta\rightarrow 0} \frac{\log(2^{n+1}/\delta)}{\log(1/\delta)d^2} +\mathcal{O}\left(\frac{\log(2^{n+1}/\delta)^{2/3}}{\log(1/\delta)d^2}\right)
\\&\le \frac{1}{d^2}\;.
\end{align*}
This concludes the proof of the complexity of \myalgo{alg-clos}.
\end{proof}
 
 After dealing with small alphabets and understanding well the differences between sequential and batch strategies for testing identity/sequential problems, one could ask whether sequential strategies have different behaviour for general alphabets when $n$ is greater than $\log(1/\delta)$. In the following, we try to transform batch algorithms to sequential ones in order to adapt to the actual $\TV$ distance.

\section{Testing uniform-the general case}\label{sec:test-unif-gen}
\cite{diakonikolas2017optimal} propose an algorithm for testing uniform in the general case. The advantage of their algorithm is that it has the tight sample  complexity (in batch setting) depending not only on $n$ but also on the error probability $\delta$. We capture the main ingredient of this article as a lemma:
\begin{lemma}\label{lemma-tes-unif}
	Let $\mathcal{D}$ a distribution on $[n]$ such that $d=\TV(\mathcal{D},U_n)>\eps$. There is a universal constant $C$ such that for all $t\ge 1 $: 
	\begin{align*}
	\mathds{E}(\TV(\tilde{\mathcal{D}}_t,U_n))\ge \mu_t(U_n)+C\min\left(\frac{d^2t^2}{n^2},d^2\sqrt{\frac{t}{n}},d\right),
	\end{align*}
	where $\mu_t(U_n)=\mathds{E}(\TV(\tilde{U}_{n,t},U_n))=\frac{1}{2} \mathds{E}_{X_1,\dots,X_t \sim U_n}\sum_{i=1}^n |\frac{1}{t}\sum_{j=1}^t 1_{X_j=i}-\frac{1}{n}|$ can be computed in $\mathcal{O}(t)$.
\end{lemma}

This Lemma (Lemma 4 from \cite{diakonikolas2017optimal}) is used along with the ideas of \mysec{bernoulli-id} to design a sequential algorithm whose complexity is on the worst case of the order of batch complexity and improves in many situations. 
At step $t$, we have samples $A_1,\dots,A_t$ from $\mathcal{D}'$. For each subset $B\subset [n]$ we denote $\tilde{\mathcal{D}'}_{t}(B)=\frac{\sum_{j=1}^t 1_{A_j\in B}}{t}$ 
and by $\mu_t(U_n)$ the expected value of $\max_{B\subset [n]} \Big|\tilde{U}_{n,t}(B)-|B|/n\Big|$. The algorithm is the following: 
\begin{algorithm}[h]
	\caption{Distinguish between  $\mathcal{D}'=U_n$ and $\TV(\mathcal{D}',U_n)>\eps$ with high probability}
	\label{algo:alg-uniform}
	\begin{algorithmic}
		\REQUIRE $A_1,\dots$ samples from $\mathcal{D}'$ 
		\ENSURE Accept if  $\mathcal{D}'=U_n$  and Reject if  $\TV(\mathcal{D}',U_n)>\eps$  with probability of error less than $\delta$
		\STATE $t=\min\{n,\sqrt{n\log(2/\delta)}\} $,  $W=1$
		\WHILE{$W=1$}
		\IF{$\exists B\subset[\lfloor n/2 \rfloor ]: \Big|\tilde{\mathcal{D}'}_t(B)-|B|/n\Big|>\max\{\phi(\delta,|B|/n,t),\phi(\delta,1-|B|/n,t)\}$ or $\TV(\tilde{\mathcal{D}'}_t,U_n)>\mu_t(U_n)+4\min\left(1,\frac{t^{3/2}}{n^{3/2}}\right)\sqrt{\frac{\log\left(\frac{2t(t+1)}{\delta}\right)}{2t}}$}
		\STATE $W=0$
		\RETURN 2
		\ELSIF{$\TV(\tilde{\mathcal{D}'}_t,U_n)<\mu_t(U_n)+C\min\left(\frac{t^2\eps^2}{n^2},\eps^2\sqrt{\frac{t}{n}},\eps\right)-4\min\left(1,\frac{t^{3/2}}{n^{3/2}}\right)\sqrt{\frac{\log\left(\frac{2t(t+1)}{\delta}\right)}{2t}}$}
		\STATE $W=0$
		\RETURN 1
		\ELSE \STATE $t=t+1$
		\ENDIF
		\ENDWHILE
	\end{algorithmic}
\end{algorithm}

The stopping times  $\tau_1$ and $\tau_2$  of \myalgo{alg-uniform} are then defined by 
\begin{alignat*}{3}
&\tau_1 &&=\inf\Bigg\{ t\ge \min\{n,\sqrt{n\log(2/\delta)}\}  :   \TV(\tilde{\mathcal{D}'}_t,U_n)<&&\mu_t(U_n)+C\min\left(\frac{t^2\eps^2}{n^2},\eps^2\sqrt{\frac{t}{n}},\eps\right) 
\\& && &&-4\min\left(1,\frac{t^{3/2}}{n^{3/2}}\right)\sqrt{\log\left(\frac{2t(t+1)}{\delta}\right)/(2t)}\Bigg\}\text{ and } 
\end{alignat*}
\begin{alignat*}{3}
\\&\tau_2 &&=\inf\Bigg\{ t\ge \min\{n,\sqrt{n\log(2/\delta)}\}  &&:    \exists B\subset[\lfloor n/2 \rfloor ]: \Big|\tilde{\mathcal{D}'}_t(B)-\frac{|B|}{n}\Big|>\max\{\phi(\delta,|B|/n,t),\phi(\delta,1-|B|/n,t)\} 
\\ &&& &&\text{ or}\TV(\tilde{\mathcal{D}'}_t,U_n)>\mu_t(U_n)+4\min\left(1,\frac{t^{3/2}}{n^{3/2}}\right)\sqrt{\log\left(\frac{2t(t+1)}{\delta}\right)/(2t)}\Bigg\}\;.
\end{alignat*}

In order to compare the sequential \myalgo{alg-uniform} with the batch one of \cite{diakonikolas2017optimal}, we need to show first that  it is indeed a $\delta$-correct algorithm and show that   its stopping times are smaller in expectation than the batch sample complexity. In the following theorem
, an upper bound of the expected stopping times is given:
\begin{theorem}  \myalgo{alg-uniform} is $\delta$-correct and its stopping times satisfy:
	
	\begin{align*}
	\mathds{E}(\tau_1(U_n))\le \max \Bigg\{ \frac{2\log(1/\delta)}{C^2\eps^2} +\frac{8}{C^2\eps^2}\log\frac{2\log(1/\delta)}{C^2\eps^2} ,\left( \frac{2n\log(1/\delta)}{C^2\eps^4} +\frac{4}{C^2\eps^4}\log\frac{2n\log(1/\delta)}{C^2\eps^4}\right)^{1/2}
	\Bigg\},
	\end{align*}
	and for $d=\TV(\mathcal{D}',U_n)>\eps$ and $B_d:=\{i: \mathcal{D}'_i>(1+d)/n\}$:
	\begin{alignat*}{2}
	&\mathds{E}(\tau_2(\mathcal{D}'))\le \min\Bigg\{&&\max \Bigg\{ \frac{3\log(1/\delta)}{C^2d^2}  ,\left( \frac{3n\log(1/\delta)}{C^2d^4} \right)^{1/2}
	\Bigg\},
	\\ & &&\frac{\log(2^{n-1}/\delta)}{\min\{\KL(|B_d|/n\pm d/2,|B_{d}|/n)\}}, \frac{\log(2^{n-1}/\delta)}{\min\{\KL(|B_{opt}|/n\pm d,|B_{opt}|/n)\}}\Bigg\}.
	\end{alignat*}
	
\end{theorem}
\begin{proof}
 
We prove here that \myalgo{alg-uniform} has an error probability less than $\delta$. The proof relies on the following uniform concentration lemma for $\TV(\tilde{\mathcal{D}}_t,U_n)$:
\begin{lemma}
	\label{lem:mc-uniform}
	\begin{align*}
	\mathds{P}\left(\exists t\ge \min\{n,\sqrt{n\log(2/\delta)}\} :  \left|\TV(\tilde{\mathcal{D}}_t,U_n)-\mathds{E}[\TV(\tilde{\mathcal{D}}_t,U_n)]\right|>4\min\left(1,\frac{t^{3/2}}{n^{3/2}}\right)\sqrt{\log\left(\frac{2t(t+1)}{\delta}\right)/(2t)}\right)\le \delta/2\;.
	\end{align*}
\end{lemma}

\begin{proof}
		If $n\le \sqrt{n\log(2/\delta)}$, we apply the union bound along with McDiarmid's inequality on $\TV(\tilde{\mathcal{D}}_t,U_n)$ which is $(1/t,\dots,1/t)$-bounded to obtain:
		\begin{align*}
		\mathds{P}\left(\exists t\ge n :  \left|\TV(\tilde{\mathcal{D}}_t,U_n)-\mathds{E}[\TV(\tilde{\mathcal{D}}_t,U_n)]\right|>4\sqrt{\log\left(\frac{2t(t+1)}{\delta}\right)/(2t)}\right)\le \sum_{t\ge n} \frac{\delta}{2t(t+1)} \;.
		\end{align*}
		If $n>\sqrt{n\log(2/\delta)}$, since the last inequality it remains to consider  $t<n$,  an application of the Bernstein form of McDiarmid’s inequality (as detailed in \cite{daskalakis2017optimal}) permits to deduce 
			\begin{align*}
		&\mathds{P}\left(\exists t\in[\sqrt{n\log(2/\delta)}, n]:  \left|\TV(\tilde{\mathcal{D}}_t,U_n)-\mathds{E}[\TV(\tilde{\mathcal{D}}_t,U_n)]\right|>4\frac{t^{3/2}}{n^{3/2}}\sqrt{\log\left(\frac{2t(t+1)}{\delta}\right)/(2t)}\right)\\
		&\le \sum_{\sqrt{n\log(2/\delta)}\le t<n} \exp{\left(\frac{-8\frac{t^2}{n^3}\log\left(\frac{2t(t+1)}{\delta}\right)}{4\frac{t^2}{n^3}+\frac{8t\sqrt{\log\left(\frac{2t(t+1)}{\delta}/2\right)}}{3n^{5/2}}}\right)}  \le \sum_{\sqrt{n\log(2/\delta)}\le t<n} \exp{\left(-\log\left(\frac{2t(t+1)}{\delta}\right)\right)}
		\\ &\le \sum_{\sqrt{n\log(2/\delta)}\le t<n} \frac{\delta}{2t(t+1)}\;.
		\end{align*}
		Combining both inequalities the lemma follows.
\end{proof}

The proof of correctness is detailed below:
\begin{itemize}
	\item If $\mathcal{D}'=U_n$, using Lemma~\ref{lem:mc-uniform}, the probability of error can be bounded as:
	\begin{align*}
	\mathds{P}\left(\tau_2\le \tau_1\right)&\le \mathds{P}\left( \exists t \ge \min\{n,\sqrt{n\log(2/\delta)}\}  , \exists B\subset[\lfloor n/2 \rfloor ]: \Big|\tilde{\mathcal{D}'}_{t}(B)-\frac{|B|}{n}\Big|>\max\{\phi(\delta,|B|/n,t),\phi(\delta,1-|B|/n,t)\} \right) \\&+\mathds{P}\left(\exists t \ge \min\{n,\sqrt{n\log(2/\delta)}\}  : \TV(\tilde{\mathcal{D}'}_t,U_n)>\mu_t(U_n)+4\min\left(1,\frac{t^{3/2}}{n^{3/2}}\right)\sqrt{\log\left(\frac{2t(t+1)}{\delta}\right)/(2t)}\right) 
	\\&\le \sum_{t\ge \min\{n,\sqrt{n\log(2/\delta)}\} , B\subset[\lfloor n/2 \rfloor ]}\mathds{P}\left( \Big|\tilde{\mathcal{D}'}_{t}(B)-\frac{|B|}{n}\Big|>\max\{\phi(\delta,|B|/n,t),\phi(\delta,1-|B|/n,t)\}\right)+\delta/2
	\\&\le \sum_{t\ge \min\{n,\sqrt{n\log(2/\delta)}\} , B\subset[\lfloor n/2 \rfloor ]}e^{-t\KL(|B|/n+\phi(\delta,|B|/n,t)|||B|/n)}+e^{-t\KL(|B|/n-\phi(\delta,1-|B|/n,t)|||B|/n)}+\delta/2
	\\&\le \delta\;.
	\end{align*}
	\item If $\TV(\mathcal{D}',U_n)>\eps$, the probability of error can be bounded as:
	\begin{align*}
	\mathds{P}\left(\tau_1\le \tau_2\right)& =\mathds{P}\Bigg(\exists t \ge \min\{n,\sqrt{n\log(2/\delta)}\}  : \TV(\tilde{\mathcal{D}'}_t,U_n)<\mu_t(U_n)+C\min\left(\frac{t^2\eps^2}{n^2},\eps^2\sqrt{\frac{t}{n}},\eps\right)\\&-4\min\left(1,\frac{t^{3/2}}{n^{3/2}}\right)\sqrt{\log\left(\frac{2t(t+1)}{\delta}\right)/(2t)}\Bigg) \\&\overset{(i)}{\le}\mathds{P}\Bigg(\exists t \ge \min\{n,\sqrt{n\log(2/\delta)}\}  : \left|\TV(\tilde{\mathcal{D}'}_t,U_n)-\mathds{E}(\TV(\tilde{\mathcal{D}'}_t,U_n))\right|\ge\\&4\min\left(1,\frac{t^{3/2}}{n^{3/2}}\right)\sqrt{\log\left(\frac{2t(t+1)}{\delta}\right)/(2t)}\Bigg) 
	\\&\overset{(ii)}{\le} \delta\;.
	\end{align*}
	where $(i)$ follows from the triangular inequality and Lemma~\ref{lemma-tes-unif} and  $(ii)$ follows from Lemma~\ref{lem:mc-uniform}.
\end{itemize}

 The sample complexity of \myalgo{alg-uniform} is given by the stopping time $\tau_1$ if the input consists of samples from the uniform distribution and by the stopping time  $\tau_2$ if the samples are from a distribution $\eps$-far from the uniform distribution.
Let's start by upper bounding $\tau_1$, the two regions related to the two stopping rules concur when 
\begin{align*}
4\min\left(1,\frac{t^{3/2}}{n^{3/2}}\right)\sqrt{\frac{\log\left(\frac{t(t+1)}{\delta}\right)}{2t}}\le \frac{C}{2}\min\left\{\frac{\eps^2t^2}{n^2},\eps^2\sqrt{\frac{t}{n}},\eps\right\},
\end{align*}
and this later condition is guaranteed by Lemma \ref{lemma-log} if 
\begin{align*}
t\ge N'_\eps:= \max \Bigg\{ \frac{2\log(1/\delta)}{C^2\eps^2} +\frac{8}{C^2\eps^2}\log\frac{2\log(1/\delta)}{C^2\eps^2} ,
\\\left( \frac{2n\log(1/\delta)}{C^2\eps^4} +\frac{4}{C^2\eps^4}\log\frac{2n\log(1/\delta)}{C^2\eps^4}\right)^{1/2}
\Bigg\}.
\end{align*}Therefore, \begin{align*}
\mathds{E}(\tau_1(U_n))\le N'_\eps.
\end{align*}
It remains to upper bound $\mathds{E}(\tau_2(\mathcal{D}'))$ for every distribution $\mathcal{D}'$ verifying $d=\TV(\mathcal{D}',U_n)>\eps$. Let $B_{opt}$ the smallest subset of $[n]$ such that $|B_{opt}|\le n/2$ and $|\mathcal{D}'(B_{opt})-U_n(B_{opt})|=d$. We make a case study on the size of $|B_{opt}|$:
\begin{itemize}
	\item If $|B_{opt}|>\sqrt{n\log(1/\delta)}$, we have :\begin{align*}
	\mathds{E}(\tau_2)&\le N'_d+\sum_{s\ge N'_d}\mathds{P}(\tau_2>s)
	\\&\le N'_d+       \sum_{t\ge N'_d-1}\mathds{P}\left( \TV(\tilde{\mathcal{D}'}_t,U_n)\le\mu_t(U_n)+4\min\left(1,\frac{t^{3/2}}{n^{3/2}}\right)\sqrt{\log\left(\frac{t(t+1)}{\delta}\right)/(2t)} \right)
	\\&\le N'_d+       \sum_{t\ge N'_d-1}\mathds{P}\Bigg( |\TV(\tilde{\mathcal{D}'}_t,U_n)-\mathds{E}(\TV(\tilde{\mathcal{D}'}_t,U_n))|\ge C\min\left\{\frac{d^2t^2}{n^2},d^2\sqrt{\frac{t}{n}},d\right\} \\&-4\min\left(1,\frac{t^{3/2}}{n^{3/2}}\right)\sqrt{\log\left(\frac{t(t+1)}{\delta}\right)/(2t)} \Bigg)
	\\&\le N'_d+       \sum_{t\ge N'_d-1}\mathds{P}\left( |\TV(\tilde{\mathcal{D}'}_t,U_n)-\mathds{E}(\TV(\tilde{\mathcal{D}'}_t,U_n))|\ge 4\min\left(1,\frac{t^{3/2}}{n^{3/2}}\right)\sqrt{\log\left(\frac{t(t+1)}{\delta}\right)/(2t)} \right)
	\\&\le N'_d+  \delta.
	\end{align*}
	\item If  $|B_{opt}|\le\sqrt{n\log(1/\delta)}$, we take $\alpha\in(0,1)$ and define  $N_\alpha$ as the minimum positive integer such that for all integers $t\ge N_\alpha$, $\max\{\phi(\delta,|B_{opt}|/n,t),\phi(\delta,1-|B_{opt}|/n,t)\}\le \alpha d.$ The existence of $N_\alpha$ is guaranteed since $\lim_{t\rightarrow \infty}\phi(B,t)=0$. We have $
	\max\{\phi(\delta,|B_{opt}|/n,N_\alpha),\phi(\delta,1-|B_{opt}|/n,N_\alpha)\}\le \alpha d$ and  $\max\{\phi(\delta,|B_{opt}|/n,N_\alpha-1),\phi(\delta,1-|B_{opt}|/n,N_\alpha-1)\}> \alpha d$ hence, 
	\begin{align*}
	\min\Bigg\{\KL\left(\frac{|B_{opt}|}{n}+\alpha d,\frac{|B_{opt}|}{n}\right),\KL\left(\frac{|B_{opt}|}{n}-\alpha d,\frac{|B_{opt}|}{n}\right)\Bigg\} \le \frac{\log\left(\frac{2^{n-1}(N-1)N}{\delta}\right)}{N-1}.
	\end{align*}
	Thus $\lim_{\delta\rightarrow 0}N =+\infty$ and from Lemma~\ref{lemma-log} we can deduce 
	\begin{align*}
	N\le \frac{\log(2^{n-1}/\delta)+4\log\left(\frac{\log(2^{n-1}/\delta)}{\min\{\KL(|B_{opt}|/n\pm\alpha d,|B_{opt}|/n)\}}\right)}{\min\{\KL(|B_{opt}|/n\pm\alpha d,|B_{opt}|/n)\}}+1\;.
	\end{align*}
	Therefore, 
	\begin{align*}\mathds{E}(\tau_2(\mathcal{D}'))&\le N_\alpha+ \sum_{t\ge N_\alpha}\mathds{P}(\tau_2(\mathcal{D}')\ge t)
	\\&\le N_\alpha+ \sum_{s\ge N_\alpha-1}\mathds{P}(|\tilde{\mathcal{D}'}_{s}(B_{opt})-|B_{opt}|/n|\le \phi(B_{opt},t))
	\\&\le N_\alpha+ \sum_{s\ge N_\alpha-1}\mathds{P}(|\tilde{\mathcal{D}'}_{s}(B_{opt})-\mathcal{D}'(B_{opt})|> d-\alpha d)\text{ (by definition of } N_\alpha)
	\\&\le N_\alpha+ \sum_{s\ge N_\alpha-1}e^{-2s((1-\alpha)d)^2} \text{ (Chernoff-Hoeffding's inequality)} 
	\\&\le N_\alpha+ \frac{e^{-2(N_\alpha-1)((1-\alpha)d)^2}}{1-e^{-2((1-\alpha)d)^2}}
	\\&\le\frac{\log(2^{n-1}/\delta)+4\log\left(\frac{\log(2^{n-1}/\delta)}{\min\{\KL(|B_{opt}|/n\pm\alpha d,|B_{opt}|/n)\}}\right)}{\min\{\KL(|B_{opt}|/n\pm\alpha d,|B_{opt}|/n)\}}+1+ \frac{e^{-2(N_\alpha-1)((1-\alpha)d)^2}}{1-e^{-2((1-\alpha)d)^2}}.
	\end{align*} After dividing by $\log(1/\delta)$, taking the limit $\delta\rightarrow0$ then $\alpha\rightarrow1$ we deduce finally: $$\limsup_{\delta\rightarrow 0} \frac{\mathds{E}(\tau_2(\mathcal{D}'))}{\log(1/\delta)}\le \frac{1}{\min\{\KL(|B_{opt}|/n\pm d,|B_{opt}|/n)\}}\underset{d\rightarrow 0}{\sim} \frac{2\frac{|B_{opt}|}{n}\left(1-\frac{|B_{opt}|}{n}\right)}{d^2}$$.
\end{itemize}
This is an interesting improvement especially when $|B_{opt}|\ll n$ but it is not sensitive to small fluctuations of $\mathcal{D}_i$ around $1/n$. For example the distribution $\{1/n+\eps,1/n,\dots,1/n,1/n-\eps\}$ whose $B_{opt}$ has size $1$ can easily transformed to a distribution with an  optimal set of size $n/2$ by adding small noise $\eta\ll \eps$ to $n/2-1$ parts among those of $1/n$ probability mass. Even though the transformed distribution has an optimal set of size $n/2$ (hence a large upper bound of the complexity ), \myalgo{alg-uniform} seems to stop on pretty the same time for both distributions. To overcome this inconvenient in this upper bound, we can use the same method to prove that for $B_d:=\{i: \mathcal{D}_i>(1+d)/n\}\subset B_{opt}$ we have: 
$$\limsup_{\delta\rightarrow 0} \frac{\mathds{E}(\tau_2(\mathcal{D}'))}{\log(1/\delta)}\le \frac{1}{\min\{\KL(|B_{d}|/n\pm d/2,|B_{d}|/n)\}}\underset{d\rightarrow 0}{\sim} \frac{8\frac{|B_{d}|}{n}\left(1-\frac{|B_{d}|}{n}\right)}{d^2}.$$
This upper bound has the advantage to count only the bigger parts of the distributions for which the noise is of the order $d/n$ at the cost of multiplying the upper bound by almost $4$.
\end{proof}

We plot in the same Figure \ref{fig:Bopt} the batch complexity and the sequential stopping time $\tau_2$ if the statistic takes into account the optimal set $B_{opt}$ and if not. It is clear that the proposed \myalgo{alg-uniform} is superior than the batch algorithm of \cite{diakonikolas2017optimal} especially when $|B_{opt}|$ is smaller than $\sqrt{n\log(1/\delta)}$ which is also exhibited in Figure \ref{fig:Bopt} for a specific example. 

We've exposed so far the advantages of sequential procedures over batch algorithms, they share the idea that if the tested distributions are far away we can hope to stop earlier. However one can wonder if there is a tangible improvement independent of distributions.
We focus on the worst case setting where we consider the eventual improvement of sequential procedures that works for all distributions. We show that we cannot hope to improve the dependency on $n$ found in the batch setting more than a constant and replacing $\eps$ by $\eps\vee \TV(\mathcal{D},U_n)$. For instance we can prove the following lower bounds for testing uniform.
 \begin{theorem}\label{theorem-lower-testing unif}
There is no stopping rule $T$ for the problem of testing $\mathcal{D}=U_n$ vs $\TV(\mathcal{D},U_n)>\eps$ with an error probability $\delta$ such that 
\begin{align*}
\mathds{P}&\left( \tau_2(T,\mathcal{D})\le  c\frac{\sqrt{n\log(1/3\delta)}}{\TV(\mathcal{D},U_n)^2}   \right) \ge 1-\delta\; \text{ if } \TV(\mathcal{D},U_n) >\eps \text{ and} 
\\\mathds{P}&\left( \tau_1(T,U_n)\le  c\frac{\sqrt{n\log(1/3\delta)}}{\eps^2}   \right) \ge 1-\delta\; ,
\end{align*}
where $c$ a universal constant. We have similar statement if we replace $\sqrt{n\log(1/3\delta)}$ by $\log(1/3\delta)$.
\end{theorem} 
This can be proven using pretty the same construction of distributions as for the batch lower bounds along with Wald's lemma. 
\begin{proof}
\label{sec:appendix-lower}

We prove only the first statement
, the others being similar. Suppose that such a stopping rule exists.
Let $d>\eps$ and $m=c\frac{\sqrt{n\log(1/3\delta)}}{d^2}$. Let $U_n$ the uniform distribution and  $D$ a uniformly chosen distribution where $D_{i}=\frac{1\pm2d}{n}$  with probability $1/2$ each. With the work of \cite{diakonikolas2016new} (Section 3), we can show that $\KL(  D^{\otimes Poi(m)}, U_n^{\otimes Poi(m)})\le C\frac{m^2d^4}{n}$ where $C$ is a constant. Therefore
\begin{align*}
\KL(  D^{\otimes m}, U_n^{\otimes m})&=m\KL(  D, U_n)
\\&=\mathds{E}(Poi(m))\KL(  D, U_n)
\\&=\KL(  D^{\otimes Poi(m)}, U_n^{\otimes Poi(m)}) \;\text{ (Wald's lemma)}
\\&\le C\frac{m^2d^4}{n}\;.
\end{align*}But \begin{align*}
\KL(  D^{\otimes m}, U_n^{\otimes m})&\ge \KL(\mathds{P}_D(\tau_2\le m),\mathds{P}_{U_n}(\tau_2\le m))
\\&\ge \KL(1-\delta,\delta)
\\&\ge \log(1/3\delta)\;,
\end{align*}
since $\mathds{P}_{D}(\tau_2\le m)\ge 1-\delta$ and $\mathds{P}_{U_n}(\tau_2\le m)=\mathds{P}_{U_n}(\tau_2\le m,\tau_1<\tau_2)+\mathds{P}_{U_n}(\tau_2\le m,\tau_1\ge \tau_2)\le \delta$.
Hence \begin{align*}
C\frac{\left(c\frac{\sqrt{n\log(1/3\delta)}}{d^2}\right)^2d^4}{n}\ge \log(1/3\delta)\;,
\end{align*}
which gives the contradiction if $c<1/\sqrt{C}$.

\end{proof}

    \begin{figure}[t!]    
    \begin{minipage}[c]{.5\linewidth}
\includegraphics[width=1\linewidth]{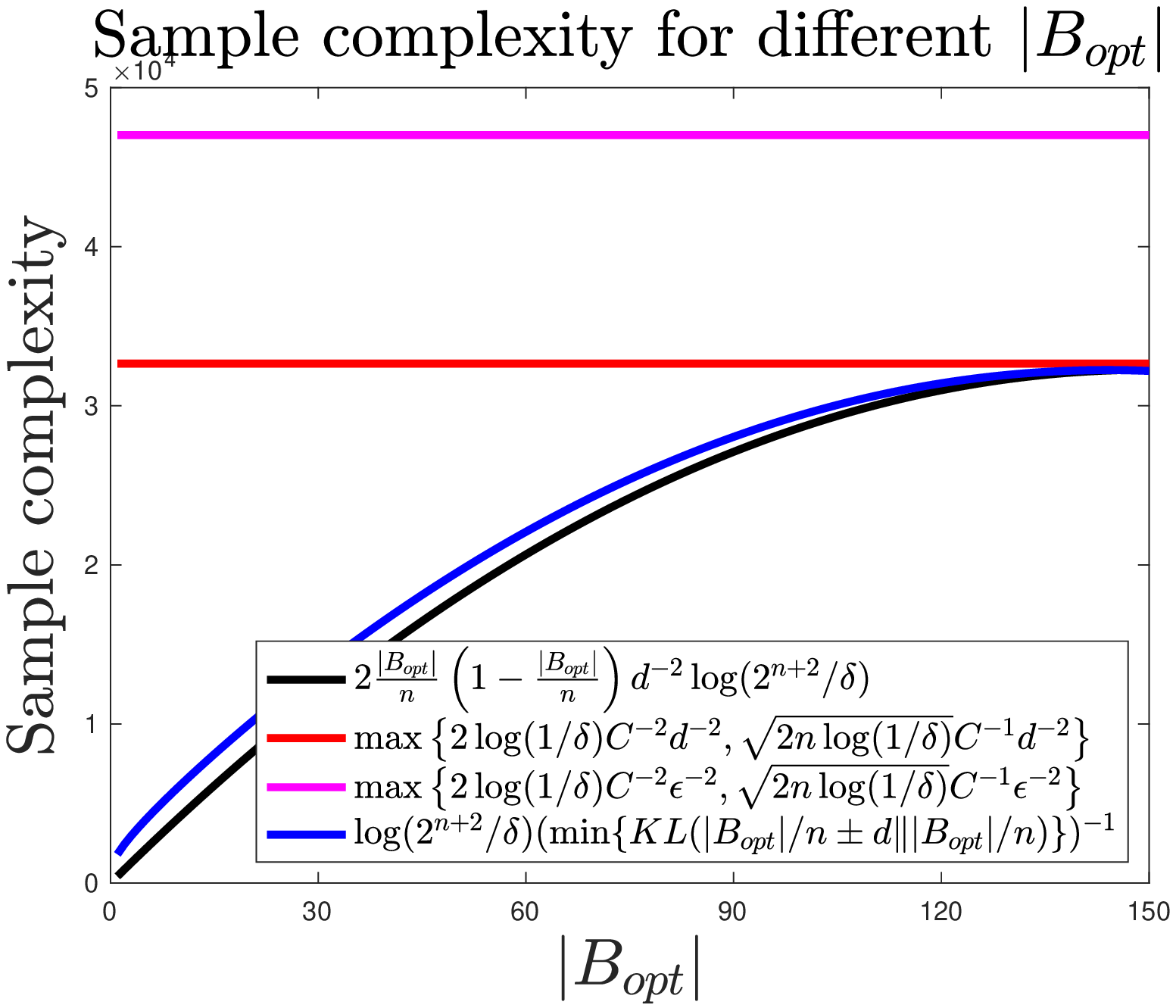}
   \end{minipage}
   \hspace*{-10pt}
\begin{minipage}[c]{.5\linewidth}
\includegraphics[width=1\linewidth]{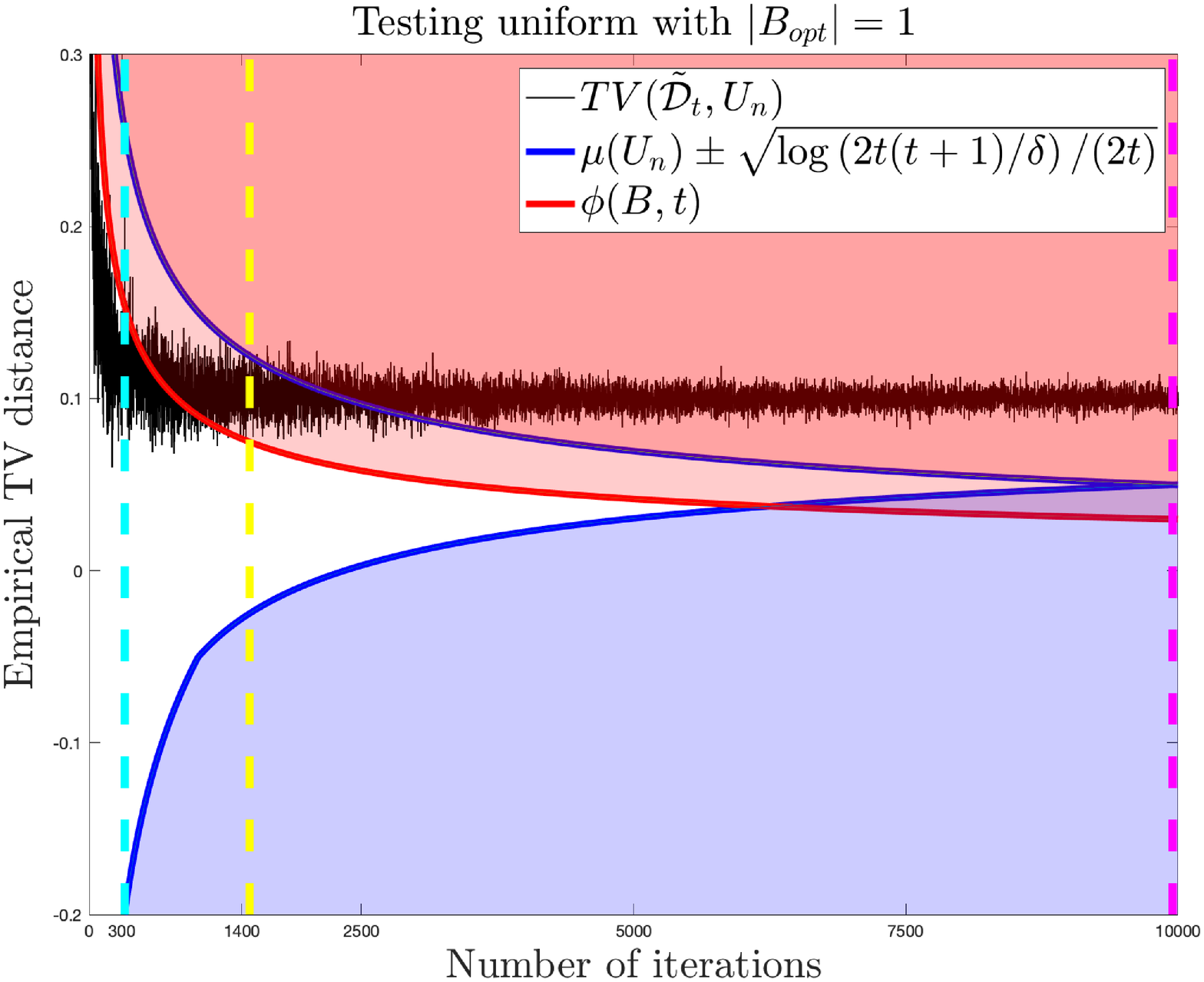}
   \end{minipage}

   \caption{Left: We compare different upper bounds on the sample complexity for testing uniform with $n=300, \eps=0.05, \delta=10^{-10}$ and $ d=\TV(\mathcal{D},U_n)=0.06$. We remark that for $|B_{opt}|<n/2$ we have better upper bounds. Right: Testing uniform for $n=10, \delta=10^{-10}, \TV(\mathcal{D}',U_n)=\eps=0.1$ and $|B_{opt}|=1$. We compare the empirical $\TV$ distance with different thresholds. The red zone corresponds to accepting $H_2$ while the blue one corresponds to accepting $H_1$. The magenta line identifies the batch threshold, the yellow line designates the sequential stopping time without looking on $B_{opt}$ and finally the cyan line defines the actual stopping time  of \myalgo{alg-uniform}.}
     \label{fig:Bopt}
\end{figure}

\section{Testing closeness-the general case}
In this section we consider testing closeness in the general case $n\ge 3$. Let us recall that we have  $\mathcal{D}_1$ and $\mathcal{D}_2$ two unknown distributions on $[n]$ and we want to distinguish $\mathcal{D}_1=\mathcal{D}_2$ and $\TV(\mathcal{D}_1,\mathcal{D}_2)>\eps$ with high probability $1-\delta$.  Inspired by the case of the Bernoulli distribution, we describe how to transform a batch algorithm to a sequential one with better expected sample complexity. 
In that case, however, identifying the sample complexity exactly remains out of reach: the dependency on $\eps$, $\delta$ and $n$ can be computed only up to a multiplicative constant.

\subsection{Batch setting}

Recently, \cite{diakonikolas2020optimal} have shown that the dependence on the error probability in the sample complexity of the closeness problem could be better than the $\log1/\delta$ found by repeating $\log1/\delta$ times the classical algorithm of \cite{chan2014optimal} and accepting or rejecting depending on the majority test.  More precisely: 

\begin{theorem}[\cite{diakonikolas2020optimal}]
\label{thm-test-clos}
$\Theta\left( \max\left( \frac{n^{2/3}\log^{1/3}(1/\delta)}{\eps^{4/3}},\frac{n^{1/2}\log^{1/2}(1/\delta)}{\eps^2},\frac{\log(1/\delta)}{\eps^2}\right)\right)$ samples are necessary and sufficient to test whether $\mathcal{D}_1= \mathcal{D}_2$ or $\TV(\mathcal{D}_1,\mathcal{D}_2)>\eps$ with an error $\delta>0$.
\end{theorem}
The main ingredient of a closeness tester is an efficient test statistic which can distinguish between the two hypothesis. Let us define by $X_i$ (resp. $Y_i$) the number of samples from $\mathcal{D}_1$ (resp. $\mathcal{D}_2$) whose values are equal to $i\in[n]$. Thinking to the $\TV$ distance we use, we could be tempted to take a decision based on the statistic $\sum_{i=1}^n |X_i-Y_i|$. However this simple statistic suffers from a principal caveat: its expected value is neither zero nor easily lower bounded when $\mathcal{D}_1=\mathcal{D}_2$. As a remedy,  \cite{diakonikolas2020optimal} propose to use the following statistic:  $Z= \sum_{i=1}^n |X_i-Y_i|+|X'_i-Y'_i|-|X_i-X'_i|-|Y_i-Y'_i|,$
where $X_i'$ and  $Y_i'$ correspond to a second set of independent samples. 
The expected value of the estimator $Z$ is obviously $0$ when the distributions $\mathcal{D}_1$ and  $\mathcal{D}_2$ are equal. On the other hand when $\TV(\mathcal{D}_1,\mathcal{D}_2)>\eps$, they provide a lower bound on the expected value of the estimator $Z$ which enable to  test closeness between $\mathcal{D}_1$ and $\mathcal{D}_2$. Since these results turn out to be similarly useful in our subsequent analysis, we summarized them in the following lemma.
\begin{lemma}[\cite{diakonikolas2020optimal}]
\label{lemma-clos} Let $d=\TV(\mathcal{D}_1,\mathcal{D}_2)$. Let $k\ge1$ and  $(k_1,k_2,k_1',k_2')\sim Multinom(4k,(1/4,1/4,1/4,1/4))$. Let $(X_i)_{i=1}^{k_1}$ and $(X'_i)_{i=1}^{k'_1}$ two sets of i.i.d. samples from $\mathcal{D}_1$ and  $(Y_i)_{i=1}^{k_2}$ and $(Y'_i)_{i=1}^{k'_2}$ two sets of i.i.d. samples from $\mathcal{D}_2$. Then there are universal constants  $c$ and $C$ such that
\begin{itemize}
\item If $\mathcal{D}_1=\mathcal{D}_2$, $\mathds{E}[Z]=0$.
\item If $\TV(\mathcal{D}_1,\mathcal{D}_2)>\eps$, $
\mathds{E}[Z] \ge C\min\left\{kd, \frac{k^2d^2}{n},\frac{k^{3/2}d^2}{\sqrt{n}} \right\}-c\sqrt{k}$. 
\end{itemize}
\end{lemma}
The lower bound on the expectation of $Z$ is obtained by a technique of Poissonization. Note that this lower bound is stronger than the one obtained for the $\textit{chi}$-square estimator; $\sum_{i=1}^n \frac{(X_i-Y_i)^2-X_i-Y_i}{X_i+Y_i}$, used by \cite{chan2014optimal}. In fact, for far distributions the lower bound on the expected value of the $\textit{chi}$-square estimator  does not allow the best dependency on $n$, $\eps$ and $\delta$. This lemma is the key ingredient behind the batch algorithm. Indeed, for sufficiently large $k=\Omega\left(\max\left( \frac{n^{2/3}\log^{1/3}(1/\delta)}{\eps^{4/3}},\frac{n^{1/2}\log^{1/2}(1/\delta)}{\eps^2},\frac{\log(1/\delta)}{\eps^2}\right)\right)$, \cite{diakonikolas2020optimal} show that  $\mathds{E}[Z] \ge C''\sqrt{k\log1/\delta}$ for a universal constant $C''$ if $\TV(\mathcal{D}_1,\mathcal{D}_2)>\eps$, then by applying McDiarmid's inequality they prove that the algorithm consisting of returning $H_2$ if $Z \ge C''\sqrt{k\log1/\delta}/2$ 
and returning $H_1$ otherwise is $\delta$-correct. In the following   we draw our inspiration from them to design a sequential algorithm for testing closeness.

\subsection{Sequential setting}
In this section we present how the sequential setting can improve the sample complexity found in the batch setting. We base our sequential tester on the same test statistic $Z$ as \cite{diakonikolas2020optimal}, but we allow the stopping rules of this new algorithm to be time-dependent. When the distributions to be tested $\mathcal{D}_1$ and $\mathcal{D}_2$ are equal, the estimator $Z_t$ cannot be very large, and if they are $\eps$-far the estimator cannot be very small: at each step, the tester compares $Z_t$ to some well chosen thresholds. If she cannot decide with sufficient confidence, she asks for more samples. This can possibly last until the two regions of decision meet. This time has the order of the complexity of the batch algorithm. 
This tester is formally defined in  \myalgo{alg-clos-gen}. For sake of simplicity, let us denote by $\Delta_t=C \min\left\{t\eps, \frac{t^2\eps^2}{n},\frac{t^{3/2}\eps^2}{\sqrt{n}} \right\}-c\sqrt{t}$ and by $\Psi_t =2\sqrt{2t\log\left(\frac{\pi^2}{3\delta}\right)+4et\log(\log(4t)+1)}$. The stopping times  $\tau_1$ and $\tau_2$  of \myalgo{alg-clos-gen} are then defined by 
\begin{align*}
\tau_1=\inf\left\{ t\ge 1 :     \left|Z_t\right|\le \Delta_t -\Psi_t\right\}, \text{ and } 
\tau_2=\inf\left\{ t\ge 1 :     \left|Z_t\right|>\Psi_t\right\}\;.
\end{align*}
We prove now that \myalgo{alg-clos-gen} is $\delta$-correct and then study its sample complexity. 
\begin{algorithm}[t]
\caption{Distinguish between  $\mathcal{D}_1=\mathcal{D}_2$ and $\TV(\mathcal{D}_1,\mathcal{D}_2)>\eps$ with high probability}
\label{algo:alg-clos-gen}
\begin{algorithmic}
\REQUIRE $A_1,\dots$ samples from $\mathcal{D}_1$ and $B_1,\dots$ samples from $\mathcal{D}_2$ 
\ENSURE Accept if  $\mathcal{D}_1=\mathcal{D}_2$  and Reject if  $\TV(\mathcal{D}_1,\mathcal{D}_2)>\eps$  with probability of error less than $\delta$
\STATE $t=1$,  $W=1$
\WHILE{$W=1$}
\STATE $(m_{1,t},m_{1,t}',m_{2,t},m_{2,t}')\sim Multinom(4t,(1/4,1/4,1/4,1/4)) $ 
 \begin{align*}Z_t= \sum_{i=1}^n |X_i-Y_i|+|X'_i-Y'_i|-|X_i-X'_i|-|Y_i-Y'_i|\;, 
\end{align*} where $X_i $(resp. $X_i',Y_i,Y_i')$ are the numbers of $i$'s in the word formed with $m_{1,t}$ (resp. $m_{1,t}',m_{2,t},m_{2,t}'$) samples from $\mathcal{D}_1$ (resp. $\mathcal{D}_1, \mathcal{D}_2,\mathcal{D}_2$) . We need only to sample the difference of $(m_{1,t}-m_{1,t-1})^+ +(m_{1,t}'-m_{1,t-1}')^+ $ from $\mathcal{D}_1$ and $(m_{2,t}-m_{2,t-1})^+ +(m_{2,t}'-m_{2,t-1}')^+ $ from $\mathcal{D}_2$.
  \IF{$\left|Z_t\right|>2\sqrt{2t\log\left(\frac{\pi^2}{3\delta}\right)+4et\log(\log(4t)+1)}$}
    \STATE $W=0$
    \RETURN 2
  \ELSIF{$\left|Z_t\right|\le C\min\left\{t\eps, \frac{t^2\eps^2}{n},\frac{t^{3/2}\eps^2}{\sqrt{n}} \right\}-c\sqrt{t} -2\sqrt{2t\log\left(\frac{\pi^2}{3\delta}\right)+4et\log(\log(4t)+1)}$}
    \STATE $W=0$
    \RETURN 1
  \ELSE \STATE $t=t+1$
  \ENDIF
\ENDWHILE
\end{algorithmic}
\end{algorithm}

\subsubsection{Correctness}

We prove here that \myalgo{alg-clos-gen} has an error probability less than $\delta$. The proof relies on the following uniform concentration lemma for $Z_t$:
\begin{lemma}
\label{lem:mcdiarmid}
For $\eta, s>1$, let  $J(\eta,s,t)= \sqrt{2\eta ts\log\left(\frac{\log(t)}{\log(\eta)}+1\right)+2t\log\left(\frac{2\zeta(s)}{\delta}\right)}$ , where $ \zeta(s)=    \sum_{n\geqslant 1} \frac{1}{n^s}$. Then 
\begin{align*}
\mathds{P}\left(\exists t\ge 1 :  \left|Z_t-\mathds{E}[Z_t]\right|>J(\eta,s,4t)\right)\le \delta\;.
\end{align*}
\end{lemma}
The proof of this lemma is inspired from \cite{howard2018uniform} and relies on dividing the set of integers into some well chosen subsets, applying union bound and finally invoking McDiarmid's inequality with specific arguments for each interval. 
Note that Lemma~\ref{lem:mcdiarmid} yields the best second order term in the complexity (up to constant factor), in contrast to a simple union bound on McDiarmid's inequality. We do not use this feature in our study of the sample complexity of the testing closeness problem as we are interested here in leading terms only (see Theorem~\ref{clos-gen-compl}). However, the $\log-\log$ dependency proves useful when showing that \myalgo{alg-clos-gen}  used with  $\eps=0$ obtains the optimal sample complexity for testing $\mathcal{D}_1=\mathcal{D}_2$ vs $\mathcal{D}_1\neq\mathcal{D}_2$ (see Theorem~\ref{thm:epszero}).
\begin{proof}
The proof uses similar arguments of \cite{howard2018uniform}.  Actually $Z_t$ is a function of $4t$ variables (the samples from the distributions) and has the property $(2,\dots,2)$-bounded differences. McDiarmid's inequality implies $\mathds{P}\left(\exists t\ge 1 :  \left|Z_t-\mathds{E}[Z_t]\right|\ge a+4bt/a\right)\le 2e^{-2b}$, taking the intervals $I_k=[\eta^k,\eta^{k+1})$ for $k$ integer we deduce for $b_k=\frac{1}{2}\log\left(\frac{2(k+1)^s}{\zeta(s)^{-1}\delta}\right)$ and  $a_k=\frac{b_k}{a_k}\eta^{k+1}$  that 
\begin{align*}
\mathds{P}\left(\exists t\ge 1 :  \left|Z_t-\mathds{E}[Z_t]\right|\ge J(\eta,s,4t)\right)&\le \sum_{k\ge 0}\mathds{P}\left(\exists t\in I_k :  \left|Z_t-\mathds{E}[Z_t]\right|\ge J(\eta,s,4t)\right)
\\&\le \sum_{k\ge 0}\mathds{P}\left(\exists t\in I_k :  \left|Z_t-\mathds{E}[Z_t]\right|\ge a_k+4b_kt/a_k\right)
\\&\le \sum_{k\ge 0} 2e^{-2b_k}
\le \sum_{k\ge 0} \delta\frac{\zeta(s)^{-1}}{(k+1)^s}
\le  \delta\;.
\end{align*}

\end{proof}


For $\eta=e$ and $s=2$, the function $J$ becomes $J(e,2,4t)=\Psi_t$ and we can use  Lemma~\ref{lem:mcdiarmid} to prove the correctness of \myalgo{alg-clos-gen} as sketched below:
\begin{itemize}
\item If $\mathcal{D}_1=\mathcal{D}_2$, using Lemma~\ref{lem:mcdiarmid}, the probability of error can be bounded as:
\begin{align*}
\mathds{P}\left(\tau_2\le \tau_1\right)\le  \mathds{P}\left(\exists t \ge 1 : \left|Z_t\right|>\Psi_t\right) \le \delta\;.
\end{align*}
\item If $\TV(\mathcal{D}_1,\mathcal{D}_2)>\eps$, the probability of error can be bounded as:
\begin{align*}
\mathds{P}\left(\tau_1\le \tau_2\right)& =\mathds{P}\left(\exists t \ge 1 : \left|Z_t\right|\le\Delta_t -\Psi_t\right) \overset{(i)}{\le}\mathds{P}\left(\exists t \ge 1 : \left|Z_t-\mathds{E}(Z_t)\right|\ge\mathds{E}(Z_t)-\Delta_t +\Psi_t\right) 
\\& \overset{(ii)}{\le} \mathds{P}\left(\exists t \ge 1 : \left|Z_t-\mathds{E}(Z_t)\right|\ge \Psi_t\right) \overset{(iii)}{\le} \delta\;.
\end{align*}
where $(i)$ follows from the triangular inequality $\left|Z_t-\mathds{E}(Z_t)\right| \geq \mathds{E}(Z_t) - Z_t$,  $(ii)$ follows by the fact that $\mathds{E}(Z_t)\geq \Delta_t$ from Lemma~\ref{lemma-clos} and  $(iii)$ follows from Lemma~\ref{lem:mcdiarmid}.
\end{itemize}



\subsubsection{Complexity}
In order to show the advantage of our sequential algorithm, we need to upper bound the expectations of the stopping times $\tau_1$ and $\tau_2$. This is done in the following theorem: 

\begin{theorem}\label{clos-gen-compl}
Let $d=\TV(\mathcal{D}_1,\mathcal{D}_2)$. The sample complexity of \myalgo{alg-clos-gen} satisfies 

\begin{itemize}
\item If $\mathcal{D}_1=\mathcal{D}_2$, $\mathds{E}(\tau_1(T,\mathcal{D}_1,\mathcal{D}_2) )\le 2N_\eps.$
\item If $\TV(\mathcal{D}_1,\mathcal{D}_2)>\eps$, $\mathds{E}(\tau_2(T,\mathcal{D}_1,\mathcal{D}_2)) \le 2N_d$.
\end{itemize}
where for all $\eta>0$, $N_\eta$ is defined by 
\begin{align*}
N_\eta=\max\Bigg\{
 \frac{128}{C^2}\frac{\log(\frac{\pi^2}{3\delta})}{\eta^2}+\frac{512e}{C^2\eta^2}\log\left(\log\left(\frac{128\log(\frac{\pi^2}{3\delta})}{\eta^2C^2}\right)+1\right)+\frac{16c^2}{C^2\eta^2},
\\ \left( \frac{128}{C^2}\frac{n^2\log(\frac{\pi^2}{3\delta})}{\eta^4}+\frac{512en^2}{C^2\eta^4}\log\left(   \log\left(   \frac{128}{C^2}\frac{n^2\log(\frac{\pi^2}{3\delta})}{\eta^4}    \right)+1    \right)+\frac{16c^2n^2}{\eta^4C^2}\right)^{1/3},
\\ \left( \frac{128}{C^2}\frac{n\log(\frac{\pi^2}{3\delta})}{\eta^4}+\frac{512en}{C^2\eta^4}\log\left(   \log\left(   \frac{128}{C^2}\frac{n\log(\frac{\pi^2}{3\delta})}{\eta^4}    \right)+1    \right)+\frac{16c^2n}{\eta^4C^2}\right)^{1/2}
\Bigg\}, 
\end{align*}
and the constants $c$ and $C$ come from Lemma~\ref{lemma-clos}.
\end{theorem}
This theorem states that $\mathcal{O}\left( \max\left( \frac{n^{2/3}\log^{1/3}(1/\delta)}{(\eps\vee \TV(\mathcal{D}_1,\mathcal{D}_2))^{4/3}},\frac{n^{1/2}\log^{1/2}(1/\delta)}{(\eps\vee \TV(\mathcal{D}_1,\mathcal{D}_2))^2},\frac{\log(1/\delta)}{(\eps\vee \TV(\mathcal{D}_1,\mathcal{D}_2))^2}\right)\right)$ samples are sufficient to distinguish between  $\mathcal{D}_1=\mathcal{D}_2$ and $\TV(\mathcal{D}_1,\mathcal{D}_2)>\eps$ with high probability.
We remark that after $N_\eps$ steps, the two stopping conditions of \myalgo{alg-clos-gen} cannot be both unsatisfied. Therefore,  the \myalgo{alg-clos-gen} stops surely before $N_\eps$ hence it has at least a comparable complexity, in the leading terms, of  the batch algorithm of \cite{diakonikolas2020optimal} when $\mathcal{D}_1=\mathcal{D}_2$. Moreover, \myalgo{alg-clos-gen} has the advantage of stopping rapidly when $\mathcal{D}_1$ and $\mathcal{D}_2$ are far away.



\begin{proof}
We start by the case $\mathcal{D}_1=\mathcal{D}_2$, we know that $\mathds{E}(\tau_1)
\le \sum_{s\le N_\eps}\mathds{P}(\tau_1\ge s)+ \sum_{s> N_\eps}\mathds{P}(\tau_1\ge s)\le N_\eps+ \sum_{s> N_\eps}\mathds{P}(\tau_1\ge s)$  so it suffices to prove that $\sum_{s> N_\eps}\mathds{P}(\tau_1\ge s)\le N_\eps$. By the definition of $\tau_1$, $\tau_1\ge s$ implies $|Z_{s-1}|>\Delta_{s-1}-\Psi_{s-1}$ but we have chosen $N_\eps$ so that if $t=s-1\ge N_\eps$, $\Delta_{s-1}-\Psi_{s-1}\ge \frac{C}{2}\min\left\{(s-1)\eps, \frac{(s-1)^2\eps^2}{n},\frac{(s-1)^{3/2}\eps^2}{\sqrt{n}} \right\}$. This last claim follows from Lemma~\ref{lemma-clos}.
 Finally    \begin{align*}
       \sum_{s> N_\eps}\mathds{P}(\tau_1\ge s) &\le  \sum_{t\ge N_\eps}\mathds{P}\left( \left|Z_t\right|>\frac{C}{2}\min\left\{t\eps, \frac{t^2\eps^2}{n},\frac{t^{3/2}\eps^2}{\sqrt{n}} \right\}\right)  
       \\&\overset{\text{(McDiarmid's inequality)}}{\le} \sum_{t\ge N_\eps-1}e^{-\frac{C^2}{16}\min\left\{t\eps^2, \frac{t^3\eps^4}{n^2},\frac{t^{2}\eps^4}{n} \right\}}  
       \le N_\eps\;.
       \end{align*}
  The last inequality is proven in~\myapp{tools}. Our claim follows.   

For the case  $d=\TV(\mathcal{D}_1,\mathcal{D}_2)>\eps$. By the definition of $\tau_2$, $\tau_2\ge s$ implies $|Z_{s-1}|\le \Psi_{s-1}$ hence by triangular inequality  $|Z_{s-1}-\mathds{E}(Z_{s-1})|\ge\mathds{E}(Z_{s-1})-\Psi_{s-1}\ge \Delta_{s-1}-\Psi_{s-1} $ therefore $|Z_{s-1}-\mathds{E}(Z_{s-1})|\ge\frac{C}{2}\min\left\{(s-1)d, \frac{(s-1)^2d^2}{n},\frac{(s-1)^{3/2}d^2}{\sqrt{n}} \right\}$ by Lemma~\ref{lemma-clos}. Hence 
    \begin{align*}
       \sum_{s> N_\eps}\mathds{P}(\tau_2\ge s) &\le  \sum_{t\ge N_\eps}\mathds{P}\left( \left|Z_t-\mathds{E}(Z_{s-1})\right|>\frac{C}{2}\min\left\{td, \frac{t^2d^2}{n},\frac{t^{3/2}d^2}{\sqrt{n}} \right\}\right)  
       \\&\overset{\text{ (McDiarmid's inequality) }}{\le} \sum_{t\ge N_d-1}e^{-\frac{C^2}{16}\min\left\{td^2, \frac{t^3d^4}{n^2},\frac{t^{2}d^4}{n} \right\}}  
       \le N_d\;.
       \end{align*}
  The later inequality is proven in~\myapp{tools}. Finally $\mathds{E}(\tau_2)
\le  N_d+ \sum_{s> N_d}\mathds{P}(\tau_2\ge s)\le 2N_d$.
 
\end{proof}
Similar to testing uniform We show that we cannot improve the dependency on $n$ found in the batch setting more than a constant and replacing $\eps$ by $\eps\vee \TV(\mathcal{D}_1,\mathcal{D}_2)$. We can prove the following lower bounds for testing closeness in the worst case setting. 
\begin{theorem}\label{theorem-lower-testing clos}
There is no stopping rule $T$ for the problem of testing $\mathcal{D}_1=\mathcal{D}_2$ vs $\TV(\mathcal{D}_1,\mathcal{D}_2)>\eps$ with an error probability $\delta$ such that 
\begin{align*}
\mathds{P}&\left( \tau_2(T,\mathcal{D}_1,\mathcal{D}_2)\le  c\frac{\sqrt{n\log(1/3\delta)}}{\TV(\mathcal{D}_1,\mathcal{D}_2)^2}   \right) \ge 1-\delta\; \text{ if } \TV(\mathcal{D}_1,\mathcal{D}_2) >\eps \text{ and} 
\\\mathds{P}&\left( \tau_1(T,\mathcal{D}_1,\mathcal{D}_2)\le  c\frac{\sqrt{n\log(1/3\delta)}}{\eps^2}   \right) \ge 1-\delta\; \; \text{ if } \mathcal{D}_1=\mathcal{D}_2,
\end{align*}
where $c$ a universal constant. We have similar statement if we replace $\frac{\sqrt{n\log(1/3\delta)}}{(\eps\vee \TV(\mathcal{D}_1,\mathcal{D}_2))^2} $ by $\frac{\log(1/3\delta)}{(\eps\vee \TV(\mathcal{D}_1,\mathcal{D}_2))^2} $ or  $\frac{n^{2/3}\log(1/3\delta)^{1/3}}{(\eps\vee \TV(\mathcal{D}_1,\mathcal{D}_2))^{4/3}} $.
\end{theorem} 
\begin{proof}
We can use the same techniques as the previous proof of Theorem~\ref{theorem-lower-testing unif}, by taking the $\KL$ between samples from $U_n\otimes U_n$  and samples from $U_n\otimes D$.
\end{proof}
On the other hand, we can deduce from Theorem~\ref{clos-gen-compl}'s proof  that with high probability we have $\tau_2\le N_{\TV(\mathcal{D}_1,\mathcal{D}_2)}$ and this upper bound has the equivalent $\mathcal{O}\left( \frac{\log\log(1/d)}{d^2}\vee\frac{n^{2/3}\log\log(1/d)^{1/3}}{d^{4/3}} \vee\frac{n^{1/2}\log\log(1/d)^{1/2}}{d^{2}} \right)$ for $d=\TV(\mathcal{D}_1,\mathcal{D}_2)\rightarrow 0$ . If we take $\eps=0$ the \myalgo{alg-clos-gen} provides stopping rules for which it does not stop if $\mathcal{D}_1=\mathcal{D}_2$  and rejects if  $\mathcal{D}_1\neq\mathcal{D}_2$ with probability at least $1-\delta$.
\begin{theorem}
\label{thm:epszero}
There is a stopping rule that can decide $\mathcal{D}_1\neq\mathcal{D}_2$ with probability at least $9/10$ using at most $\mathcal{O}\left( \frac{\log\log(1/d)}{d^2}\vee\frac{n^{2/3}\log\log(1/d)^{1/3}}{d^{4/3}} \vee\frac{n^{1/2}\log\log(1/d)^{1/2}}{d^{2}} \right)$ samples where $d=\TV(\mathcal{D}_1,\mathcal{D}_2)$. 
\end{theorem}
 This improves the results of \cite{daskalakis2017optimal} where the dependency in $n$ is $n/\log n$. Furthermore, we cannot find stopping rules whose sample complexity is tighter than this upper bound as stated in the following theorem. 
 \begin{theorem}\label{theorem-lower-testing =vs neq}
There is no stopping rule $T$ for the problem of testing $\mathcal{D}_1=\mathcal{D}_2$ vs $\mathcal{D}_1\neq\mathcal{D}_2$ with an error probability $1/16$ such that 
\begin{align*}
\mathds{P}\left( \tau_2(T,\mathcal{D}_1,\mathcal{D}_2)\le  C\frac{n^{1/2}\log\log(1/d)^{1/2}}{d^{2}}   \right) \ge \frac{15}{16}\;,
\end{align*}
where  $d=\TV(\mathcal{D}_1,\mathcal{D}_2)$ and $C$ a universal constant. We have similar statements if we replace $\frac{n^{1/2}\log\log(1/d)^{1/2}}{d^{2}}$ by $\frac{\log\log(1/d)}{d}$ or $\frac{n^{2/3}\log\log(1/d)^{1/3}}{d^{4/3}}$.
\end{theorem} 
To sum up, 
a number $\Theta\left( \frac{\log\log(1/d)}{d^2}\vee\frac{n^{2/3}\log\log(1/d)^{1/3}}{d^{4/3}} \vee\frac{n^{1/2}\log\log(1/d)^{1/2}}{d^{2}} \right)$ of samples is necessary and sufficient to decide whether $\mathcal{D}_1=\mathcal{D}_2$ or $\mathcal{D}_1\neq\mathcal{D}_2$  with probability $15/16$.
\begin{proof}
We use ideas similar to \cite{karp2007noisy}. We prove only the first statement, the others being similar. 
Let's start by a lemma:
\begin{lemma}\label{lem-kl}
Let $X$  and $Y$ two random variables and $E$ some event verifying $\mathds{P}_X(E)\ge 1/3$ and $\mathds{P}_Y(E)<1/3$, we have 
\begin{align*}
\KL(\mathds{P}_X,\mathds{P}_Y)\ge -\frac{1}{3}\log(3\mathds{P}_Y( E))-\frac{1}{e}.
\end{align*}
\end{lemma}
\begin{proof}
By data processing property of Kullback-Leibler’s divergence: 
\begin{align*}
\KL(\mathds{P}_X,\mathds{P}_Y)&\ge \KL(\mathds{P}_X(E),\mathds{P}_Y(E))
\\&\ge \mathds{P}_X(E)\log\frac{\mathds{P}_X(E)}{\mathds{P}_Y(E)}+(1-\mathds{P}_X(E))\log\frac{1-\mathds{P}_X(E)}{1-\mathds{P}_Y(E)}
\\&\ge -\frac{1}{3}\log(3\mathds{P}_Y( E))+(1-\mathds{P}_X(E))\log(1-\mathds{P}_X(E))
\\&\ge-\frac{1}{3}\log(3\mathds{P}_Y( E))-\frac{1}{e}\;.
\end{align*}
\end{proof}
Suppose by contradiction that there is a stopping rule such that 
\begin{align*}
\mathds{P}\left( \tau_2(T,\mathcal{D}_1,\mathcal{D}_2)> \frac{n^{1/2}\log\log(1/d)^{1/2}}{Cd^{2}}   \right) \le  \frac{1}{16}\;,
\end{align*} whenever $d=\TV(\mathcal{D}_1,\mathcal{D}_2)>0$. Let $\eps_1=1/3$, we construct recursively $T_k=\left\lceil \frac{n^{1/2}\log\log(1/\eps_k)^{1/2}}{C\eps_k^{2}} \right\rceil =\frac{C'\sqrt{n}}{\eps_{k+1}^2}$ where $C$ and $C'$ are constants defined later. For each integer $j$, we take $m_j\sim Poi(j)$. Let $U_n$ the uniform distribution and  $D_k$ a uniformly chosen distribution where $D_{k,i}=\frac{1\pm2\eps_k}{n}$  with probability $1/2$ each. With the work of \cite{diakonikolas2016new} (Section 3), we can show that $\KL(   U_n^{\otimes m_j}\otimes D_k^{\otimes m_j},U_n^{\otimes m_j}\otimes U_n^{\otimes m_j})\le C''\frac{j^2\eps_k^4}{n}$ where $C''$ is a constant.   Since $\TV(U_n,D_k)=\eps_k>0$, $\mathds{P}\left( \tau_2(T,U_n,D_k)> T_k \right) \le 1/16$. Let $E_k$ be the event that the stopping rule decides that the distributions are not equal between $T_{k-1}$ and $T_k$. We have $\mathds{P}\left( \tau_2(T,U_n,D_{k})\le T_{k-1} \right) \le 1/3$ since otherwise  
Lemma~\ref{lem-kl} implies:
\begin{align*}
-\frac{1}{3}\log\left(3\mathds{P}\left( \tau_2(T,U_n,U_n)\le T_{k-1} \right) \right)-\frac{1}{e}&\le \KL( U_n^{\otimes m_{T_{k-1}}}\otimes D_k^{\otimes m_{T_{k-1}}},U_n^{\otimes m_{T_{k-1}}}\otimes U_n^{\otimes m_{T_{k-1}}} )\\&\le C''\frac{T_{k-1}^2\eps_k^4}{n}
\\&\le C''C'\;,
\end{align*}thus 
\begin{align*}
    \mathds{P}\left( \tau_2(T,U_n,U_n)\le T_{k-1} \right) \ge e^{-3C''C'-3/e}/3 >0.1,
\end{align*}
for good choice of $C'$ and this contradicts the fact the the stopping rule is infinite with a probability at least $0.9$. The stopping rule is $0.1$ correct so  $\mathds{P}\left( \tau_2(T,U_n,D_{k})< +\infty \right) \ge 0.9$ then
\begin{align*}
    \mathds{P}\left( T_{k-1}<\tau_2(T,U_n,D_{k})\le T_{k} \right) \ge 0.9-1/3-1/16>0.5.
\end{align*} The same inequalities for the Kullback-Leibler’s divergence as above permits to deduce: 
\begin{align*}
1\ge \sum_{k\ge 1}\mathds{P}\left( T_{k-1}<\tau_2(T,U_n,U_n)\le T_k \right)& \ge \sum_{k\ge 1}\frac{1}{3}e^{-3C''T_k^2\eps_k^4/n-3/e}
\\&\ge \sum_{k\ge 1}\frac{1}{3e^2}e^{-3C''/C^2 \log\log(1/\eps_k)} \text{ and choosing }C \text{ st } 3C''/C^2=1/2
\\&\ge\sum_{k\ge 1} \frac{1}{3e^2}\frac{1}{\sqrt{\log(1/\eps_k)}}\;.
\end{align*}But the later sum is divergent because if we denote $a_k=\log(1/\eps_k)$, we have $a_{k+1}\le a_k+\frac{1}{4}\log\log a_k +\mathcal{O}(1)$ thus $a_k=\mathcal{O}(k\log\log k)$ therefore $\frac{1}{\sqrt{\log(1/\eps_k)}}\ge \frac{c}{k }$ which is divergent.
\end{proof}


\begin{remark}\label{rem:doubling}
We note that we can transform the batch algorithms to sequential ones using 
the doubling search technique. For instance, we can use the algorithm of  \cite{diakonikolas2020optimal} as a black box and test sequentially for $1\le t\le \log(1/\eps)$ whether $\cD_1=\cD_2$ or $\TV(\cD_1,\cD_2)>2^{-t}$ with a probability of failure no more than $\delta_t= \delta/t^2$. If at some step the batch-algorithm rejects we reject and halt, otherwise we accept.  Note that one could think, at first sight, that this reduction even permits to estimate $\TV(\cD_1,\cD_2)$; it is however not the case, since we cannot ensure for different distributions of TV distance less than $\eps$ that the proposed algorithm will respond the right answer (the black box algorithm can return any hypothesis $H_1$ or $H_2$ if the $\TV$ distance is strictly between $0$ and $\eps$) hence the doubling search algorithm (as described) cannot be used for tolerant testing.
On the other hand, the doubling search method can lead to the same order of sample complexity as Theorem~\ref{clos-gen-compl} for testing closeness problem.
Nevertheless, when it comes to multiplicative constants, the two algorithms appear to have significantly different behaviors. In all experiments the actual sample complexity of \myalgo{alg-clos-gen} appears to be better by an important constant factor.  This does not show off in the bounds (since the multiplicative constants are not known), but this can be understood at least when the TV distance $d$ satisfies $\eps< 2^{-k-1}< d = 2^{-k}(1-\eta)  < 2^{-k}$, when the doubling search obviously requires $\approx 4$ times more samples than necessary. Actually, even without this discretization effect, the difference is significant. This sub-optimality of the doubling search algorithm can be seen clearly for small alphabets where we can characterize the sample complexity to the constant: We gain a factor $4$ using our approach while using doubling search algorithm requires up to $4$ times more than the batch sample complexity when the $\TV$ distance is strictly greater than $\eps$. 
We thank the unknown reviewer and Clément L Canonne for bringing up the doubling search technique.
\end{remark}






\section{Conclusion}
We have provided a tight analysis of the complexity of testing identity and closeness  for small $n$, where the importance of sequential procedures is clearly exhibited. 
\newline
For the general case, we proposed tight algorithms for testing identity and closeness where the complexity can depend on the actual $\TV$ distance between the two distributions. 
We note that for some specific families of distributions the improvement can be much more than the general one. This is the case of distributions  concentrated in small sets which can be tested rapidly by  sequential strategies.

\textbf{Acknowledgements.}\\
Omar Fawzi acknowledges the support of the European Research Council (ERC Grant Agreement No. 851716).
Aurélien Garivier acknowledges the support of the Project IDEXLYON of the University of Lyon, in the framework of the Programme Investissements d’Avenir (ANR-16-IDEX-0005), and Chaire SeqALO (ANR-20-CHIA-0020).




\appendix

\section{General lower bounds and their proofs}\label{sec:proofs-low-ber}

In this section we present lower bounds for testing identity and testing closeness in the general case of $n\geq 2$ and provide the proofs of the lower bounds presented in the paper. 

\subsection{ Lower bound for testing identity in the general case $n\geq 2$}
\label{sec:proof-lowerbound-kl-id}
We first provide and prove a lower bound result for testing identity.
\begin{lemma}
Let $\mathcal{D} $ be a known distribution on $[n]$. Let $T$ a stopping rule for testing identity: $\mathcal{D}'= \mathcal{D}$ vs ${\TV(\mathcal{D}',\mathcal{D})>\eps}$ with an error probability $\delta$. Let  $\tau_1$ and $\tau_2$ the associated stopping times. We have 
\begin{itemize}
\item $\mathds{E}(\tau_1(T,\mathcal{D})) \ge \frac{\log1/3\delta}{\inf_{\mathcal{D}'' \text{s.t.}\TV(\mathcal{D}'',\mathcal{D})>\eps } \KL(\mathcal{D} , \mathcal{D}'')}$ \text{ if } $\mathcal{D}'= \mathcal{D}$.
\item $\mathds{E}(\tau_2(T,\mathcal{D}'))\ge \frac{\log1/3\delta} {\KL(\mathcal{D}' , \mathcal{D})}$    \text{ if }${\TV(\mathcal{D}',\mathcal{D})>\eps}$.
\end{itemize}
\end{lemma}

\begin{proof}
We consider the two different cases $\mathcal{D}'=\mathcal{D}$ and $\TV(\mathcal{D}',\mathcal{D})>\eps$.
\paragraph{The case $\mathcal{D}'=\mathcal{D}$.}
We denote by $\mathds{P}_\mathcal{D}$  the probability distribution on $[n]^{\mathds{N}}$ with independent marginals $X_i$ of distribution $\mathcal{D}$.  Let $Z=(X_1,\dots ,X_{\tau_1})$ and $\mathcal{D}''$ be a distribution such that  $\TV(\mathcal{D}'',\mathcal{D})>\eps$.
The data processing property of the Kullback-Leibler divergence implies
\begin{align}\label{proof-kll-id}
\KL\left(\mathds{P}_{\mathcal{D}}^Z , \mathds{P}_{\mathcal{D}''}^Z\right)\ge \KL\left(\mathds{P}_{\mathcal{D}}(\tau_1<\infty) , \mathds{P}_{\mathcal{D}''}(\tau_1<\infty)\right) \;. 
\end{align}
But $\mathds{P}_{\mathcal{D}}(\tau_1<\infty) \ge 1-\delta$ and $\mathds{P}_{\mathcal{D}''}(\tau_1<\infty)\le \delta $.  Moreover, $x\rightarrow \KL(x , y)$ is increasing on $(y,1)$ and $y\rightarrow \KL(x , y)$ is decreasing on $(0,x)$ hence $\KL\left(\mathds{P}_{X}(E) , \mathds{P}_{Y}(E)\right) \ge \KL(1-\delta , \delta)$. Tensorization property and Wald's lemma (\ref{Wald-lemma}) lead to
\begin{align*}
\KL\left(\mathds{P}_{\mathcal{D}}^Z , \mathds{P}_{\mathcal{D}''}^Z\right)=\mathds{E}(\tau_1(T,\mathcal{D}))\KL(\mathcal{D} , \mathcal{D}'')\;.
\end{align*} 
The inequality \ref{proof-kll-id} becomes 
\begin{align*}
\mathds{E}(\tau_1(T,\mathcal{D} )\KL(\mathcal{D} , \mathcal{D}'')\ge \KL(1-\delta , \delta)\ge \log1/3\delta\;,
\end{align*}
which is valid for all distribution $\mathcal{D}''$, consequently  
\begin{align*}
\mathds{E}(\tau_1(T,\mathcal{D})\ge\frac{\log1/3\delta}{ \inf_{\mathcal{D}'': \TV(\mathcal{D},\mathcal{D}'')>\eps }  \KL(\mathcal{D}, \mathcal{D}'')}\;.
\end{align*}

\paragraph{The case $\TV(\mathcal{D}',\mathcal{D})>\eps$.}
With similar notations and techniques we find for $Z=(X_1,\dots ,X_{\tau_2})$
\begin{align*}\mathds{E}(\tau_2(T,\mathcal{D}' )\KL(\mathcal{D}' , \mathcal{D})&=\KL\left(\mathds{P}_{\mathcal{D}'}^Z , \mathds{P}_{\mathcal{D}}^Z\right)
\\&\ge \KL\left(\mathds{P}_{\mathcal{D}'}(\tau_2<\infty) , \mathds{P}_{\mathcal{D}}
(\tau_2<\infty)\right) 
\\&\ge \KL(1-\delta, \delta)
\\&\ge\log1/3\delta\;.
\end{align*}
Finally we can deduce
\begin{align*}
\mathds{E}(\tau_2(T,\mathcal{D}' ))\ge\frac{\log1/3\delta}{ \KL(\mathcal{D}', \mathcal{D})}\;.
\end{align*}
\end{proof}

\subsection{ Lower bound for testing closeness in the general case $n\geq 2$} 
\label{sec:proof-lowerbound-kl}
We propose the following lower bounds for testing closeness in general case 
\begin{lemma}\label{lowerbound-kl}
Let $T$ a stopping rule for testing $\mathcal{D}_1= \mathcal{D}_2$ vs ${\TV(\mathcal{D}_1,\mathcal{D}_2)>\eps}$ with an error probability $\delta$. Let  $\tau_1$ and $\tau_2$ the associated stopping times. We have 
\begin{itemize}
\item $\mathds{E}(\tau_1(T,\mathcal{D}_1,\mathcal{D}_2)) \ge \frac{\log1/3\delta}{\inf_{\mathcal{D}'_{1,2} \text{s.t.}\TV(\mathcal{D}_1',\mathcal{D}'_2)>\eps } \KL(\mathcal{D}_1 , \mathcal{D}'_1)+\KL(\mathcal{D}_2 , \mathcal{D}'_2)}$ \text{ if } $\mathcal{D}_1= \mathcal{D}_2$.
\item $\mathds{E}(\tau_2(T,\mathcal{D}_1,\mathcal{D}_2))\ge \frac{\log1/3\delta}{\inf_{\mathcal{D} }  \KL(\mathcal{D}_1 , \mathcal{D})+ \KL(\mathcal{D}_2 , \mathcal{D})}$    \text{ if }${\TV(\mathcal{D}_1,\mathcal{D}_2)>\eps}$.
\end{itemize}
\end{lemma}
\begin{proof}
Similarly as in the previous proof, we consider the two different cases $\mathcal{D}'=\mathcal{D}$ and $\TV(\mathcal{D}',\mathcal{D})>\eps$.
\paragraph{The case $\mathcal{D}_1=\mathcal{D}_2$.}
We denote by $\mathds{P}_{\mathcal{D}_1,\mathcal{D}_2}$
the probability distribution on $([n]\times [n])^{\mathds{N}}$ with independent marginals $(X_i,Y_i)$ of distribution $\mathcal{D}_1\otimes \mathcal{D}_2$.  Let $Z=(X_1,Y_1\dots ,X_{\tau_1},Y_{\tau_1})$.  Let $\mathcal{D}'_1,\mathcal{D}'_2$ be two distributions such that  $\TV(\mathcal{D}'_1,\mathcal{D}'_2)>\eps$. Data processing property of Kullback-Leibler’s divergence implies
\begin{align}\label{proof-kll}
\KL\left(\mathds{P}_{\mathcal{D}_1,\mathcal{D}_2}^Z, \mathds{P}_{\mathcal{D}'_1,\mathcal{D}'_2}^Z\right)\ge \KL\left(\mathds{P}_{\mathcal{D}_1,\mathcal{D}_2}(\tau_1<\infty), \mathds{P}_{\mathcal{D}'_1,\mathcal{D}'_2}(\tau_1<\infty)\right) \;. 
\end{align}
By definition of $\tau_1$ we have $\mathds{P}_{\mathcal{D}_1,\mathcal{D}_2}(\tau_1<\infty)\ge 1-\delta$ and $\mathds{P}_{\mathcal{D}'_1,\mathcal{D}'_2}(\tau_1<\infty)\le \delta$.
Tensorization property and Wald's lemma (\ref{Wald-lemma}) lead to
\begin{align*}
\KL\left(\mathds{P}_{\mathcal{D}_1,\mathcal{D}_2}^Z, \mathds{P}_{\mathcal{D}'_1,\mathcal{D}'_2}^Z\right)=\mathds{E}(\tau_1(T,\mathcal{D}_1,\mathcal{D}_1))\KL(\mathcal{D}_{1} , \mathcal{D}'_1)+\mathds{E}(\tau_1(T,\mathcal{D}_1,\mathcal{D}_2))\KL(\mathcal{D}_{2}  , \mathcal{D}'_2)\;.
\end{align*} 
The inequality \ref{proof-kll} becomes 
\begin{align*}
\mathds{E}(\tau_1(T,\mathcal{D}_1,\mathcal{D}_2))\KL(\mathcal{D}_{1} , \mathcal{D}'_1)+\mathds{E}(\tau_1(T,\mathcal{D}_1,\mathcal{D}_2))\KL(\mathcal{D}_{2}  , \mathcal{D}'_2)\ge \KL(1-\delta , \delta)\ge \log1/3\delta\;,
\end{align*}
which is valid for all distribution $\mathcal{D}'_1$ and $\mathcal{D}_2'$ such that $\TV(\mathcal{D}'_1,\mathcal{D}'_2)>\eps$, consequently  
\begin{align*}
\mathds{E}(\tau_1(T,\mathcal{D}_1,\mathcal{D}_2))\ge\frac{\log1/3\delta}{ \inf_{\mathcal{D}'_{1,2} \text{s.t. } \TV(\mathcal{D}'_1,\mathcal{D}'_2)>\eps  }  \KL(\mathcal{D}_{1} , \mathcal{D}'_1)+\KL(\mathcal{D}_{2}  , \mathcal{D}'_2)}\;.
\end{align*}

\paragraph{The case $\TV(\mathcal{D}_1,\mathcal{D}_2)>\eps$.}
Likewise we prove for $Z=(X_1,Y_1\dots ,X_{\tau_2},Y_{\tau_2})$ and $\mathcal{D}$ a distribution on $[n]$. 
\begin{align*}
\mathds{E}(\tau_2(T,\mathcal{D}_1,\mathcal{D}_2))\KL(\mathcal{D}_{1} , \mathcal{D})+\mathds{E}(\tau_2(T,\mathcal{D}_1,\mathcal{D}_2))\KL(\mathcal{D}_{2}  , \mathcal{D})&=\KL\left(\mathds{P}_{\mathcal{D}_1,\mathcal{D}_2}^Z, \mathds{P}_{\mathcal{D},\mathcal{D}}^Z\right)
\\&\ge \KL\left(\mathds{P}_{\mathcal{D}_1,\mathcal{D}_2}(\tau_2<\infty) , \mathds{P}_{\mathcal{D},\mathcal{D}}
(\tau_2<\infty)\right) 
\\&\ge \KL(1-\delta, \delta)
\\&\ge\log1/3\delta\;.
\end{align*} 
which is valid for all distribution $\mathcal{D}$, consequently  
\begin{align*}
\mathds{E}(\tau_2(T,\mathcal{D}_1,\mathcal{D}_2))\ge\frac{\log1/3\delta}{ \inf_{\mathcal{D} }  \KL(\mathcal{D}_1 , \mathcal{D})+\KL(\mathcal{D}_2 , \mathcal{D})}.
\end{align*}
\end{proof}


\section{Technical lemmas}

\subsection{Kullback-Leibler divergence}\label{KL-sec}
\begin{definition}[Kullback Leibler divergence]
The Kullback Leibler divergence is defined for two distributions $p$ and $q$ on $[n]$ as 
\begin{align*}
\KL(p,q)=\sum_{i=1}^n p_i\log\left(\frac{p_i}{q_i}\right)\;.
\end{align*}We denote by $\KL(p,q)=\KL(\mathcal{B}(p),\mathcal{B}(q))$.

\end{definition}
Kullback-Leibler’s divergence satisfies data-processing  and tensorization properties: 
\begin{proposition}
 Let $p,p',q$ and $q'$ distributions on $[n]$, we have 
 \begin{itemize}
     \item \textbf{Non negativity} $\KL(p,q)\ge0$. 
     \item \textbf{Data processing} Let $X$ a random variable and $g$ a function. Define the random variable $Y=g(X)$,  we have 
\begin{align}
\KL\left(p^X,q^X\right)\ge \KL\left(p^Y,q^Y\right).
\end{align}
    \item \textbf{Tensorization}  \begin{align*}
\KL(p\otimes p',q\otimes q')=\KL(p,q)+\KL(p',q').
\end{align*} 
 \end{itemize}
 
\end{proposition}

\subsection{Poissonization}
The Poisson law of parameter $\lambda$ is denoted $Poi(\lambda)$ and defined as follows. 
\begin{align*}
\forall k \in \mathds{N},  \space ~~~~ \mathds{P}(Poi(\lambda)=k) = \frac{\lambda^k}{k!}e^{-\lambda}\;.
\end{align*}
Poisson law is important for the analysis of testing' algorithms. In fact, some important random variables becomes independent when we take a number of samples following a Poisson law. 
\begin{lemma}[Poissonization]\label{poisson}
Let $k\sim Poi(\tau)$ and $X=(X_1,\dots,X_k)$ i.i.d samples from a distribution $p$ on $[n]$. For $i\in [n]$, we denote $Y_i$ the number of times $i$ appears in the tuple  $X$. We have 
\begin{enumerate}
\item $\{Y_1,\dots,Y_n\}$ are independent. 
\item For all $i\in [n]$, $Y_i\sim Poi(\tau p_i)$.
\end{enumerate}
\end{lemma}
\subsection{Wald's lemma}

\begin{lemma}[\cite{wald1944cumulative}]\label{Wald-lemma}
Let $(X_n)_{n\ge0}$ i.i.d random variables and $N\in \mathds{N}$ a random variable independent of $(X_n)_n$. Suppose that $N$ and $X_1$ have finite expectations. we have   
\begin{align*}
\mathds{E}(X_1+\dots+X_N)=\mathds{E}(N)\mathds{E}(X_1)\;.
\end{align*}
\end{lemma}

\subsection{Tools for non asymptotic inequalities}
\label{sec:tools}
We group here different lemmas that help us to deal with the kl-divergence or logarithmic relations in order to find non asymptotic results. 
We start by giving some useful lemmas for the Kullback-Leibler’s divergence  between Bernoulli variables.
\begin{lemma}[Lemmas for kl-divergence.]\label{lemma-kl} Let $q>p$ two numbers in $[0,1]$. Then
\begin{itemize}
\item $2(p-q)^2\le \KL(p , q)\le \frac{(p-q)^2}{q(1-q)},$
\item $\KL(p , q)\underset{q\rightarrow p}{\sim} \frac{(p-q)^2}{2q(1-q)},$
\item $\KL(q , p)=\int_p^q du\int_p^u dv \frac{1}{v(1-v)}.$ 
\end{itemize}

\end{lemma}
\paragraph{Sketch of proof.} The LHS of the first inequality is Pinsker's inequality, the RHS can be proven using the inequality $\log(1+x)\le x$, the second equivalence can be found by developing the $\log$ function and the third equality is proven by calculating the integral.

\begin{lemma}\label{lemma-kl2}[Developing kl]Let $q, \eps $ and $\alpha$ positive real numbers such that $q+\eps<1$ and $\alpha<1$, we have for $\alpha$ close enough to $1$
\begin{align*}
\frac{1}{\KL(q+\alpha\eps , q)}\le \frac{1}{\KL(q+\eps , q)}+(1-\alpha)\sup_{[q,q+\eps]}\frac{1}{x(1-x)}\;.
\end{align*}

\end{lemma}
\begin{proof}
We use the inequality $\frac{1}{1-x}\le 1+2x$ for $0<x<1/2$.  We write 
\begin{align*}
\frac{1}{\KL(q+\alpha\eps , q)}= \frac{1}{\KL(q+\eps , q)(1-x)}\;,
\end{align*}
where $x=\frac{\KL(q+\eps , q)-\KL(q+\alpha\eps , q)}{\KL(q+\eps , q)}<\frac{1}{2}$ if $\alpha$ is close enough to $1$. Hence 
\begin{align*}
\frac{1}{\KL(q+\alpha\eps , q)}&\le \frac{1}{\KL(q+\eps , q)(1-x)}
\\&\le \frac{1}{\KL(q+\eps , q)}(1+2x)
\\&\le \frac{1}{\KL(q+\eps , q)}+2 \frac{\KL(q+\eps , q)-\KL(q+\alpha\eps , q)}{\KL(q+\eps , q)^2}
\\ &\le \frac{1}{\KL(q+\eps , q)}+ \frac{2}{\KL(q+\eps , q)^2}\int_{q+\alpha\eps}^{q+\eps} du\int_q^u dv \frac{1}{v(1-v)}
\\ &\le \frac{1}{\KL(q+\eps , q)}+ \frac{2(1-\alpha)\eps^2}{\KL(q+\eps , q)^2}\sup_{[q,q+\eps]}\frac{1}{v(1-v)}
\\ &\le \frac{1}{\KL(q+\eps , q)}+ \frac{2(1-\alpha)\eps^2}{2\eps^2}\sup_{[q,q+\eps]}\frac{1}{v(1-v)}
\\ &\le \frac{1}{\KL(q+\eps , q)}+ (1-\alpha)\sup_{[q,q+\eps]}\frac{1}{v(1-v)}\;.
\end{align*}
\end{proof}
When we deal with inequalities involving $t$ and $\log t$ (or $\log\log t$) and want to deduce inequalities only on $t$, the following lemma proves to be useful.  
\begin{lemma}\label{lemma-log} Let $t,a>1$ and $b$ real numbers. We have the following implications: 
\begin{itemize}
\item If $b\ge a+1 :$\begin{align*}
t\ge b+2a\log(b) \Rightarrow t\ge b+a\log(t)\;,
\end{align*}
\item If $b\ge 1:$\begin{align*}t\ge b-a\log(t) \Rightarrow t\ge b-a\log(b)\;,
\end{align*}
\item If $b\ge 2a : $ \begin{align*}
t\ge b+2a\log(\log(b)+1) \Rightarrow t\ge b+a\log(\log(t)+1)\;.
\end{align*}
\end{itemize}
\end{lemma}
\begin{proof}We prove only the first statement, the others being similar.
Let $f(t)=t-b-a\log(t)$, we have $f'(t)=1-a/t$ thus $f$ is increasing on $(a,+\infty)$. Let $t\ge b+2a\log(b)>a$, 
\begin{align*}
f(t)\ge f(b+2a\log(b))&=b+2a\log(b)-b-a\log(b+2a\log(b))
\\&=a\log(b)-a\log(1+2a\log(b)/b))
\\&\ge a\log(1+a)-a\log(1+2ab/eb) \ \ \text{ because } \log(b)\le b/e
\\&\ge 0\;.
\end{align*}
\end{proof}
For instance, by applying this lemma, we can obtain:
\begin{lemma}\label{lemma_general-case}
Recall the definition of $N_\eta$: \begin{align*}
N_\eta=\max\Bigg\{
 \frac{128}{C^2}\frac{\log(\frac{\pi^2}{3\delta})}{\eta^2}+\frac{512e}{C^2\eta^2}\log\left(\log\left(\frac{128\log(\frac{\pi^2}{3\delta})}{\eta^2C^2}\right)+1\right)+\frac{16c^2}{C^2\eta^2},
\\ \left( \frac{128}{C^2}\frac{n^2\log(\frac{\pi^2}{3\delta})}{\eta^4}+\frac{512en^2}{C^2\eta^4}\log\left(   \log\left(   \frac{128}{C^2}\frac{n^2\log(\frac{\pi^2}{3\delta})}{\eta^4}    \right)+1    \right)+\frac{16c^2n^2}{\eta^4C^2}\right)^{1/3},
\\ \left( \frac{128}{C^2}\frac{n\log(\frac{\pi^2}{3\delta})}{\eta^4}+\frac{512en}{C^2\eta^4}\log\left(   \log\left(   \frac{128}{C^2}\frac{n\log(\frac{\pi^2}{3\delta})}{\eta^4}    \right)+1    \right)+\frac{16c^2n}{\eta^4C^2}\right)^{1/2}
\Bigg\} \;.
\end{align*}
Let  $\eta>0$, if $t\ge N_\eta$, then 
\begin{align*}
\min\left\{ t\eta, \frac{t^2\eta^2}{n},\frac{t^{3/2}\eta^2}{\sqrt{n}}\right\}\ge \frac{4}{C}\sqrt{2t\log\left(\frac{\pi^2}{3\delta}\right)+4et\log(\log(t)+1)} +\frac{2c}{C}\sqrt{t}\;.
\end{align*}\end{lemma} 
Finally, the next lemma shows that the complexity of \myalgo{alg-clos-gen} cannot exceed $N_{d\vee\eps}$ very much.
\begin{lemma}We have for all $d>0$:
$\sum_{t\ge N_d}e^{-\frac{C^2}{16}\min\left\{td^2, \frac{t^3d^4}{n^2},\frac{t^{2}d^4}{n}\right\}}\le N_d$.
\end{lemma}
\begin{proof}We have
\begin{align*}
\sum_{t\ge N_d}e^{-\frac{C^2}{16}\min\left\{td^2, \frac{t^3d^4}{n^2},\frac{t^{2}d^4}{n}\right\}}&\le \sum_{t\ge nd^{-2}}e^{-\frac{C^2}{16}td^2 }+\sum_{n\ge t\ge N_d-1}e^{-\frac{C^2}{16}\frac{t^3d^4}{n^2} }+\sum_{nd^{-2}> t> n}e^{-\frac{C^2}{16}\frac{t^{2}d^4}{n} }
          \\&\le \sum_{t\ge nd^{-2}}e^{-\frac{C^2}{16}td^2 }+\sum_{n\ge t\ge N_d-1}e^{-2C^{1/3}\frac{td^{4/3}}{n^{2/3}} }+\sum_{nd^{-2}> t> n}e^{-\frac{C}{2}\frac{td^2}{\sqrt{n}}}
          \\&\le \frac{1}{1-e^{-\frac{C^2}{16}d^2 }}+\frac{1}{1-e^{-2C^{1/3}\frac{d^{4/3}}{n^{2/3}} }}+\frac{1}{1-e^{-\frac{C}{2}\frac{d^2}{\sqrt{n}}}}
          \\&\le \frac{32}{C^2d^2}+\frac{n^{2/3}}{C^{1/3}d^{4/3}}+\frac{4\sqrt{n}}{Cd^2} \ \ \text{ since } 1-e^{-x} \ge x/2 \text{ for  }0<x<1 
          \\&\le N_d\;.
\end{align*}
\end{proof}
\end{document}